%% file: arxiv.tex
\documentclass[pra,twocolumn,amsmath,amssymb,superscriptaddress,notitlepage]{revtex4-2}

\usepackage{graphicx}
\usepackage{float}                 
\usepackage{siunitx}               
\usepackage{mathtools}
\usepackage{physics}               
\usepackage{amsthm}
\usepackage[utf8]{inputenc}        

\newtheorem{theorem}{Theorem}
\newtheorem{corollary}{Corollary}
\newtheorem{lemma}{Lemma}
\newtheorem{remark}{Remark}

\usepackage{enumitem}
\usepackage{bm}
\usepackage[english]{babel}
\usepackage[autostyle, english = american]{csquotes} \MakeOuterQuote{"}

\usepackage{hyperref}
\usepackage{xurl}                  
\hypersetup{
	colorlinks = true,
	urlcolor   = blue,
	linkcolor  = red,
	citecolor  = red
}
\usepackage[capitalize]{cleveref}  

\usepackage{xcolor}
\newcommand{\blue}[1]{#1}
\newcommand{\bluemath}[1]{#1}

\usepackage{tikz}
\usetikzlibrary{arrows}            
\usetikzlibrary{calc}              
\tikzstyle{process} = [rectangle, rounded corners, minimum width=3cm, minimum height=1cm, text centered, draw=black, align= center]
\tikzstyle{virtualprocess} = [rectangle, rounded corners, minimum width=3cm, minimum height=1cm, text centered, draw=black, dashed, align=center]
\tikzstyle{emptyprocess} = [rectangle, minimum width=3cm, minimum height=1cm, text centered, align=center]
\tikzstyle{point} = [circle, inner sep=0pt, minimum size=3pt, fill=black]
\tikzstyle{arrow} = [thick, ->, >=stealth]
\tikzstyle{dashedarrow} = [thick, ->, >=stealth, dashed]

\newcommand{\deltaone}{\delta_1}
\newcommand{\deltatwo}{\delta_2}

\newcommand{\epsATb}{\varepsilon_\mathrm{dep-1}}
\newcommand{\epsATc}{\varepsilon_\mathrm{dep-2}}
\newcommand{\epssrc}{\varepsilon_\mathrm{ind}}
\newcommand{\epspe}{\varepsilon_\mathrm{PE}}
\newcommand{\epssp}{\varepsilon_\mathrm{sp}}
\newcommand{\epspa}{\varepsilon_\mathrm{PA}}
\newcommand{\epsev}{\varepsilon_\mathrm{EV}}
\newcommand{\epssec}{\varepsilon_\mathrm{sec}}
\newcommand{\epscorr}{\varepsilon_\mathrm{corr}}

\newcommand{\Sindep}{S_{0,0}}
\newcommand{\Sdep}{S_{\delta_1,\delta_2}}

\newcommand{\zpovm}{\mathcal{Z}}
\newcommand{\xpovm}{\mathcal{X}}
\newcommand{\vecnX}{\vec n_\xpovm}

\newcommand{\Prindep}{\Pr_{\Sindep}}
\newcommand{\Prdep}{\Pr_{\Sdep}}
\newcommand{\Prboth}{\Pr_{S}}

\newcommand{\Findep}{\mathcal{F}_{0,0}}
\newcommand{\Eindep}{\mathcal{E}_{0,0}}
\newcommand{\Edep}{\mathcal{E}_{\delta_1,\delta_2}}
\newcommand{\Eboth}{\mathcal{E}}
\newcommand{\Eindepdecoy}{\Eindep^{\mathrm{decoy}}}
\newcommand{\Findepdecoy}{\Findep^{\mathrm{decoy}}}

\newcommand{\etadet}{\eta_{\text{det} }}
\newcommand{\dcprob}{d_{\text{det} }}
\newcommand{\etachar}{\Delta_{\eta}}
\newcommand{\dcchar}{\Delta_{\mathrm{dc}}}

\newcommand{\nph}{n_\mathrm{ph}}
\newcommand{\eph}{e_\mathrm{ph}}
\newcommand{\lambdaec}{\lambda_\mathrm{EC}}

\newcommand{\affvqcc}{Vigo Quantum Communication Center, University of Vigo, Vigo E-{36310}, Spain}
\newcommand{\affuvigo}{Escuela de Ingeniería de Telecomunicación, Department of Signal Theory and Communications, University of Vigo, Vigo E-36310, Spain}
\newcommand{\affatlantic}{atlanTTic Research Center, University of Vigo, Vigo E-36310, Spain}
\newcommand{\affwaterloo}{Institute for Quantum Computing and Department of Physics and Astronomy, University of Waterloo, Waterloo, Ontario, Canada, N2L 3G1}

\begin{document}

\setlength{\parskip}{3pt}
\setlength{\parindent}{0pt}

\author{Guillermo Currás-Lorenzo}
\affiliation{\affvqcc} \affiliation{\affuvigo} \affiliation{\affatlantic}
\author{Margarida Pereira}
\affiliation{\affvqcc} \affiliation{\affuvigo} \affiliation{\affatlantic}
\author{Shlok Nahar}
\affiliation{\affwaterloo}
\author{Devashish Tupkary}
\affiliation{\affwaterloo}

\title{Security of quantum key distribution with source and detector imperfections through phase-error estimation}

\begin{abstract}
Quantum key distribution promises information-theoretic security, but practical implementations face vulnerabilities due to device imperfections. Existing security proofs typically address source and detector imperfections separately and cannot be directly combined, while the few attempts to handle both simultaneously have assumed idealized scenarios with infinitely many protocol rounds, which is insufficient to certify the security of real-world systems. Here, we close this gap by introducing a modular lifting theorem that extends any (present or future) phase-error-based security proof for BB84-type protocols to incorporate detector efficiency mismatches. By applying this theorem to a state-of-the-art analysis for source imperfections, we obtain the first finite-key security proof against general coherent attacks that simultaneously addresses both source and detector imperfections.
\end{abstract}

\maketitle

\input{main-body}

\clearpage
\appendix

\crefname{section}{Appendix}{Appendices}
\Crefname{section}{Appendix}{Appendices}

\setcounter{figure}{0}
\renewcommand{\thefigure}{A\arabic{figure}}
\providecommand{\theHfigure}{\thefigure}
\renewcommand{\theHfigure}{A\arabic{figure}}

\input{supp-body}

\bibliography{refs,refs_extra}

\end{document}

%% file: main-body.tex
\section*{Introduction}

Quantum key distribution (QKD) promises unconditional security guaranteed by the fundamental laws of quantum mechanics~\cite{bennettQuantumCryptography1984}. However, practical QKD implementations suffer from device imperfections that deviate from the idealized theoretical models typically assumed in security analyses, invalidating their guarantees~\cite{xuSecureQuantum2020,ImplementationAttacksBSI}. This represents one of QKD's most critical challenges~\cite{NSA,NSA_europea}, and has led to growing interest in developing security proofs that incorporate such imperfections  \cite{gottesmanSecurityQuantum2004,zapateroSecurityQuantum2021,pereiraModifiedBB842023,tamakiLosstolerantQuantum2014,curras-lorenzoSecurityFramework2025,sixtoQuantumKey2025,curras-lorenzoNumericalSecurity2025,fungSecurityProof2009,lydersenSecurityQuantum2010,trushechkinSecurity2022,bochkovSecurity2019,grasselliQuantumKey2025,tupkaryPhaseError2025,maroy2010security,sunSecurityQuantum2021,marcominiLosstolerantQuantum2025,arqandMutualInformation2024,marwahProvingSecurity2024,naharImperfectDetectors2025,kaminRenyiSecurity2025}, most of which are based on phase-error estimation \cite{gottesmanSecurityQuantum2004,zapateroSecurityQuantum2021,pereiraModifiedBB842023,tamakiLosstolerantQuantum2014,curras-lorenzoSecurityFramework2025,sixtoQuantumKey2025,curras-lorenzoNumericalSecurity2025,fungSecurityProof2009,lydersenSecurityQuantum2010,trushechkinSecurity2022,bochkovSecurity2019,grasselliQuantumKey2025,tupkaryPhaseError2025,maroy2010security,sunSecurityQuantum2021,marcominiLosstolerantQuantum2025}.

\blue{Despite these advances, a critical limitation persists in the existing literature: lack of modularity. Security proofs addressing source imperfections \cite{zapateroSecurityQuantum2021,pereiraModifiedBB842023,tamakiLosstolerantQuantum2014,curras-lorenzoSecurityFramework2025,sixtoQuantumKey2025,curras-lorenzoNumericalSecurity2025} and those addressing detector imperfections \cite{fungSecurityProof2009,lydersenSecurityQuantum2010,trushechkinSecurity2022,bochkovSecurity2019,grasselliQuantumKey2025,tupkaryPhaseError2025} have been developed largely independently, and their results cannot be directly combined. This creates a significant gap, as practical QKD implementations inevitably suffer from both types of imperfections simultaneously. While a few previous works have attempted to prove security in the presence of both~\cite{maroy2010security,sunSecurityQuantum2021,marcominiLosstolerantQuantum2025}, these have focused only on the asymptotic regime, i.e., under the unrealistic assumption that the protocol is executed for an infinitely large number of rounds, and thus cannot be applied to certify the security of real-world systems. Moreover, these works suffer from additional technical limitations (see \textit{Discussion}). Outside phase-error estimation-based proofs, there is limited work addressing detector imperfections \cite{naharImperfectDetectors2025}. Thus, an approach capable of proving security in the presence of both source and detector imperfections, in the finite-key regime and against general attacks, has been missing.}

\blue{Here, we close this critical gap by establishing a modular lifting theorem that takes any existing phase-error-based security proof assuming an ideal active-basis-choice BB84 receiver and extends it to incorporate bounded detector efficiency mismatches. Our result builds on Ref.~\cite{tupkaryPhaseError2025}, generalizing this result from a security proof for a specific scenario—BB84 with an ideal source—into a modular lifting theorem, thereby significantly expanding its scope and utility.} 

\blue{As a concrete demonstration, we apply our lifting theorem to the state-of-the-art phase-error bound in Ref.~\cite{curras-lorenzoSecurityFramework2025}, which addresses general source imperfections in BB84-type protocols. By doing so, we obtain the first complete finite-key security proof against general coherent attacks that simultaneously incorporates both source and detector imperfections. Beyond this specific application, the main strength of our framework lies in its modularity: future researchers can focus on developing improved techniques for source imperfections, and our lifting theorem will immediately extend those advances to incorporate detector imperfections, without any additional theoretical work.}

\section*{Results}

\subsection*{Phase-error estimation for generalized BB84-type scenarios.}

In an ideal BB84 protocol, for each round, Alice chooses a random bit and basis, and emits a perfect qubit Pauli eigenstate corresponding to her choice. Ref.~\cite{tupkaryPhaseError2025} shows how to extend the standard phase-error-based security proof for this scenario to cover detection efficiency mismatches. In practice, however, due to source imperfections, the states emitted by Alice will not have this ideal form. Here, we describe a class of generalized BB84-type scenarios in which Alice's source may be imperfect, and specify how a security proof based on phase-error estimation can be formulated for such scenarios under the assumption that Bob's measurement setup satisfies the basis-independent detection efficiency condition. Then, we explain how our result can be directly applied to extend any such security proof to cover detection efficiency mismatches on Bob's setup. For concreteness, in the main text, we focus on single-photon BB84-type protocols; for the application of our techniques to decoy-state BB84-type protocols, see Methods.

In a generalized BB84-type scenario, Alice selects a sequence of setting choices $j_1...j_N$ with probability $p_{j_1...j_N}$, prepares some global state $\ket{\psi_{j_1...j_N}}_{a_1 ... a_N}$, and sends it to Bob through the quantum channel. Note that this model covers scenarios in which Alice's global state may be arbitrarily correlated between rounds~\footnote{Often, security proofs assume that both Alice's setting choices and the states she emits are IID, i.e., $p_{j_1...j_N} = p_{j_1}...p_{j_N}$ and $\ket{\psi_{j_1...j_N}}_{a_1 ... a_N} = \ket{\psi_{j_1}}_{a_1} ...\ket{\psi_{j_N}}_{a_N}$; however, our result is general and could be applied to proofs that do not require such IID assumptions on Alice's state preparation, i.e., given a phase-error-estimation strategy for non-IID imperfect sources and perfect detectors (Eq. (2)), our result can upgrade it to a result for non-IID imperfect sources and imperfect detectors. Even more generally, the global state $\rho_{j_1...j_N}$ emitted by Alice could be mixed. In this case, one can simply consider that Alice prepares a purification $\ket{\psi_{j_1...j_N}}_{S_1...S_N a_1 ... a_N}$ of $\rho_{j_1...j_N}$, where the ancillary shield systems $S_1 ... S_N$ are not accessible to Eve.}. Typically, $j_k \in \{0,1,2,3\}$, but Alice can more generally emit a different number of states, see e.g., the three-state protocol \cite{boileauUnconditionalSecurity2005,tamakiLosstolerantQuantum2014}. 
As for Bob, for each round, he probabilistically chooses either the $\zpovm \coloneqq \{\Gamma_0^{(\zpovm)},\Gamma_1^{(\zpovm)},\Gamma_\bot^{(\zpovm)}\}$ or $\xpovm \coloneqq \{\Gamma_0^{(\xpovm)},\Gamma_1^{(\xpovm)},\Gamma_\bot^{(\xpovm)}\}$ positive operator-valued measure (POVM) and performs the selected measurement. \blue{Here, $\Gamma_0^{(\beta)}$ and $\Gamma_1^{(\beta)}$ correspond to detection events for bit values 0 and 1, respectively, while $\Gamma_\bot^{(\beta)}$ corresponds to no detection (discarded round).} Then, he announces which rounds were detected, and his basis choice $\beta \in \{\zpovm,\xpovm\}$ for each detected round. The basis-independent detection efficiency assumption requires that the POVM element corresponding to discarded rounds is the same for both POVMs, i.e, $ \Gamma_\bot^{(\zpovm)} = \Gamma_\bot^{(\xpovm)}$. This only holds when the efficiency and dark count rates of all of Bob's detectors are identical \footnote{Note that in an active polarization-encoded BB84 detector setup, Bob's actual POVMs actually have four elements: no click, click in detector 0, click in detector 1, and double-click event. To define a three-outcome POVM, double-click events must be assigned to a bit value. The basis-independent detection efficiency condition requires not only identical detectors but also random bit assignment for double-click events. Other implementations would require some similar post-processing to a three-outcome POVM.}. Under this assumption, Bob's measurement admits the following equivalent description (in the sense of being the same quantum to classical map): (1) apply a basis-independent filter $\{\tilde F, \mathbb{I} - \tilde F\}$ that determines whether or not he obtains a detection, (2) for the rounds that pass the filter, choose a measurement basis $\beta \in \{\zpovm,\xpovm\}$, and (3) apply a two outcome POVM $\{G_0^{(\beta)},G_1^{(\beta)}\}$ that determines his bit value. For more information about this equivalent description, see \cref{app:bobs_measurement}.

After Bob's measurements, Alice announces certain sifting information that allows the users to determine which detected rounds will be used for key generation. For concreteness, we assume here that the sifted key is extracted from the rounds in which Alice's setting choice satisfies $j \in \{0,1\}$ (typically referred to as $Z$-basis emissions) and Bob chooses $\zpovm$.

To prove security, we consider a scenario in which Alice generates a source-replacement state
\begin{equation}
\label{eq:source_replacement_state}
     \ket{\Psi} = \sum_{j_1...j_N} \sqrt{p_{j_1...j_N}} \ket{j_1 ... j_N}_{A_1  ... A_N} \ket{\psi_{j_1  ... j_N}}_{a_1  ... a_N}.
\end{equation}
Clearly, if Alice measures her ancilla systems $A_1...A_N$ in the computational basis, this scenario is equivalent to the actual protocol. We can also consider an equivalent scenario in which Alice and Bob first learn which rounds will be used for key generation, and only afterwards learn their bit values. For this,  Bob applies the filter $\{\tilde F, \mathbb{I} - \tilde F\}$ to all rounds, and then he makes his basis choice $\beta$ for the rounds that pass the filter (i.e., the detected rounds) and announces this information, but does not learn his bit value yet. Meanwhile, for the detected rounds in which Bob announces $\zpovm$, Alice attempts a projection onto the subspace spanned by $\{\ket{0}_{A},\ket{1}_{A}\}$; if the projection is successful, then the users learn this is a key generation round. 

The security of the key can then be related to the number of \textit{phase} errors that Alice and Bob would obtain if, at this point, for each key generation round, Alice measured in the $\{\ket{+}_A,\ket{-}_A\}$ basis (rather than the  $\{\ket{0}_A,\ket{1}_A\}$ basis) and Bob measured $\{G_0^{(\xpovm)},G_1^{(\xpovm)}\}$ (rather than $\{G_0^{(\zpovm)},G_1^{(\zpovm)}\}$). In particular, a security analysis based on phase-error estimation requires proving a statement of the form
\begin{equation}
\label{eq2:guarantee_BIDE}
\Prindep(\bm{e_\mathbf{ph}} > 	\Eindep(\bm{\vec n_\xpovm},\bm{n_K})) \leq \epssrc^2,
\end{equation}    
which can then be directly used to determine the length of the final output key and its security parameter, see \cref{eq:key_length} below. Here, \blue{the subscript $\Sindep$ under $\Pr$ and the subscript 0,0 under $\Eboth$ indicate that the bound holds for a phase-error estimation protocol in which Bob's measurement satisfies the basis-independent detection efficiency condition. Also, }$\bm{e_\mathrm{ph}}\coloneqq \bm{n_\mathrm{ph}}/\bm{n_K}$, where $\bm{n_\mathrm{ph}}$ is the random variable (RV) associated to the number of phase errors in the scenario described above, $\bm{n_K}$ is the RV associated to the number of key-generation rounds, $\bm{\vec n_{\xpovm}}$ is the random vector containing all the announced statistics from rounds in which Bob announces $\xpovm$, $\Eindep$ is a function relating these quantities, and $\epssrc^2$ is the failure probability of the bound\blue{, which is related to the secrecy parameter of the final key (see key-length formula in \cref{eq:key_length})}. Note that throughout the manuscript, we use bold letters to denote classical RVs.

\subsection*{Main result}

In this work, we show that, given a statement of the form in \cref{eq2:guarantee_BIDE}, one can extend it to the case of detection efficiency mismatch, obtaining a statement of the form
\begin{equation}
\label{eq2:objective}
\Prdep\left(\bm{e_\mathbf{ph}} > 	\Edep\big(\bm{\vec n_\xpovm}, \bm{n_K}\big)\right) \leq \epssrc^2+\epsATb^2+\epsATc^2.
\end{equation}
Here, the subscript $\Sdep$ under $\Pr$ indicates that the bound holds for a phase-error estimation protocol in which Bob's measurement setup suffers from a detection efficiency mismatch parameterized by $\delta_1$ and $\delta_2$. These parameters are introduced in Ref.~\cite{tupkaryPhaseError2025} and quantify how far off Bob's measurement setup is from satisfying the basis-independent detection efficiency assumption. For the general definition of these parameters, and of the phase-error estimation protocol in the case of detector efficiency mismatch, see \cref{app:bobs_measurement}. \blue{Also, the terms $\varepsilon_{\text{dep-1}}^2$ and $\varepsilon_{\text{dep-2}}^2$ are failure probabilities arising from the concentration inequalities used to bound the effects of the detection efficiency mismatch, and contribute to the secrecy parameter of the final key (see key-length formula in \cref{eq:key_length}).}

In the general case, the extended bound is defined as
\begin{equation}
\label{eq:Edep1}
\begin{aligned}
\Edep(\bm{\vec n_\xpovm},\bm{n_K})
&=
\bluemath{\frac{1}{\bm{n_K}}\,
\underset{n \in \mathcal{W}_{\delta_2}(\bm{n_K})}{\max}\;
n\,\Eindep(\bm{\vec n_\xpovm},n)}
\\
&+
\frac{\delta_1+\gamma^{\epsATb}_{\mathrm{bin}}(\bm{n_K},\delta_1)}
{1-\delta_2-\gamma^{\epsATc}_{\mathrm{bin}}(\bm{n_K},\delta_2)},
\end{aligned}
\end{equation}
where
\begin{equation}
    \mathcal{W}_{\delta_2}(m) \coloneqq \mathcal{N}_{\leq} \left(\left\lfloor\frac{m}{1-\delta_2-\gamma^{\epsATc}_{\mathrm{bin}}(m, \delta_2)} \right \rfloor\right)
\end{equation}
with $\mathcal{N}_{\leq}(m) \coloneqq \{0,1,...,m\}$ denoting the set of non-negative integers up to $m$; and $\gamma^{\varepsilon}_{\mathrm{bin}}$ is a finite-size deviation term \blue{defined as}
\begin{equation}
\begin{aligned}
\label{eq:gammabin_defined}
&\bluemath{\gamma^{\varepsilon}_{\mathrm{bin}}(n, \delta) \\
&:= \min\left\{x \geq 0 : \sum_{i = \lfloor n(\delta+x) \rfloor}^n \binom{n}{i} \delta^i (1-\delta)^{n-i} \leq \varepsilon^2\right\}.}
\end{aligned}
\end{equation}

While this formula requires an optimization that may be computationally intensive, it can be simplified in most cases. In particular, if the function
\begin{equation}
\label{eq:Findepdef}
\Findep\big(\bm{\vec n_\xpovm}, \bm{n_{K}}\big) \coloneqq \bm{n_K} \Eindep\big(\bm{\vec n_\xpovm}, \bm{n_{K}}\big), 
\end{equation}
is non-decreasing with respect to $\bm{n_K}$---which typically holds since the concentration inequalities used in QKD security proofs typically converge sublinearly---then the optimization term $\max_{n \in \mathcal{W}_{\delta_2}(\bm{n_K})}  n \, \Eindep (\bm\vecnX,n)$ in \cref{eq:Edep1} can be simply replaced by $  \bm{n^*} \, \Eindep (\bm\vecnX,\bm{n^*})$, where 
$$\bm{n^*}= \left \lfloor\frac{\bm{n_K}}{1-\delta_2-\gamma^{\epsATc}_{\mathrm{bin}}(\bm{n_K}, \delta_2)}\right \rfloor.$$

Moreover, if $\Eindep$ is also non-increasing with respect to $\bm{n_K}$---as is typical for most security proofs, since concentration inequalities become tighter as the sample size increases---then we can obtain the even simpler expression
\begin{equation}
\begin{aligned}
\label{eq:Edep2}
 &\Edep(\bm{\vec n_\xpovm},\bm{n_{K}})     \\
 &=\frac{\Eindep\big(\bm{\vec n_\xpovm}, \bm{n_{K}}\big) + \deltaone + \gamma^{\epsATb}_{\mathrm{bin}}(\bm{n_{K}}, \deltaone)}{1-\delta_2-\gamma^{\epsATc}_{\mathrm{bin}}(\bm{n_{K}}, \delta_2)}.   
\end{aligned}
\end{equation}
For the proof of the statements in this subsection, see \cref{thm:maintheorem,maincorollary} in \cref{app:proof_main_result}.

More informally, our result implies that, asymptotically,
\begin{equation}
    \eph^{(\Sdep)} \lesssim \frac{\eph^{(\Sindep) } + \delta_1}{1-\delta_2},
\end{equation}
where $\eph^{(\Sindep)}$ denotes the phase-error rate in the case of basis-independent detection efficiency (or an upper bound on it), and $\eph^{(\Sdep)}$ denotes the phase-error rate in the case of a detection efficiency mismatch parameterized by $\delta_1$ and $\delta_2$.

\blue{We remark, however, that our result applies in the finite-key regime and finite-key effects are explicitly incorporated into the bound $\Edep$. These include both the corrections from the original proof (through $\Eindep$) and those introduced by our lifting theorem (through the $\gamma^{\varepsilon}_{\mathrm{bin}}$ terms in \cref{eq:Edep1,eq:Edep2}).} \blue{We also emphasize the modular nature of our result: if applied to a security proof that incorporates source imperfections (such as Ref.~\cite{curras-lorenzoSecurityFramework2025}), these are automatically inherited through the original bound $\Eindep$, while our lifting theorem incorporates the detector imperfections independently into the extended bound $\Edep$. Our result is agnostic to how $\Eindep$ is derived, it only requires that such a bound exists.}

\subsection*{Secret-key length}

Given a general phase-error-rate bound of the form in \cref{eq2:guarantee_BIDE} or \cref{eq2:objective}, i.e.,
\begin{equation}
\label{eq:eph_bound_both}
\Prboth(\bm{e_\mathbf{ph}} > 	\mathcal{E}(\bm{\vec n_\xpovm},\bm{n_K})) \leq \epspe^2,
\end{equation}
security can be established using either entropic uncertainty relations (EUR) and the leftover hashing lemma (LHL) \cite{tomamichelUncertaintyRelation2011,tomamichelTightFinitekey2012,tomamichelLargelySelfcontained2017} or phase-error correction arguments \cite{koashiSimpleSecurity2009}. As demonstrated in \cite{tupkaryPhaseError2025,curras-lorenzoTightFinitekey2021} (see also \cite{hayashiConciseTight2012,kawakamiSecurity}), a bound of this form directly yields finite-key security against general attacks in the variable-length framework \cite{ben-orUniversalComposable2005}, where the output key length depends on observations during protocol execution. Specifically, under the EUR+LHL framework, one can achieve a secure key length \cite[Theorem 1]{tupkaryPhaseError2025}
\begin{equation}
\label{eq:key_length}
\begin{aligned}
    \bm{l} &= \bm{n_K} \big[1-h\big(\mathcal{E}(\bm{\vec n_\xpovm},\bm{n_K})\big)\big] - \bm\lambdaec \\&- 2 \log(1/2 \epspa) - \log(2/\epsev),
\end{aligned}
\end{equation}
that is $(\epscorr+\epssec)$-secure, with $\epscorr = \epsev$ and $\epssec = 2 \epspe + \epspa$. Here, $\epsev$ is the failure probability of the error verification step, $\epspe$ is the square root of the failure probability of the phase-error rate bound (see \cref{eq:eph_bound_both}), $\epspa>0$ can be freely chosen, and $\bm\lambdaec$ denotes the number of bits leaked during error-correction, which must be a function of the announced protocol observations.

\subsection*{Numerical results}

To demonstrate the utility of our result, we apply it to the security proof in Ref.~\cite{curras-lorenzoSecurityFramework2025}, which is based on phase-error estimation. This analysis is valid in the finite-key regime and against general attacks, and can incorporate imperfect and partially characterized sources. Specifically, it only requires the knowledge that the state $\rho_j^{(k)}$ emitted by Alice in any round $k$ when she selects encoding $j$, satisfies
\begin{equation}
\label{eq:fidelity_constraint}
    \ev{\rho_j^{(k)}}{\phi_j} \geq 1-\epsilon,
\end{equation}
where $\ket{\phi_j}$ is a known reference state and $\epsilon\geq 0$ bounds its deviation to $\rho_j^{(k)}$. For concreteness, we consider a BB84-type scenario in which Alice's setting choices are $j\in \{0_Z,1_Z,0_X,1_X\}$, and the reference states are qubit states of the form
\begin{equation}		\label{eq:qubit_state_model}
			 \ket{\phi_{j}} =	\cos(\theta_j)\ket{0_Z}+\sin(\theta_j)\ket{1_Z},
		\end{equation}
		where {$\theta_j = (1+\delta_{\rm SPF}/\pi) \varphi_j/2$,} $\varphi_j\in\{0,\pi,\pi/2,3\pi/2\}$ for $j\in\{0_Z,1_Z,0_X,1_X\}$, and $\delta_{\rm SPF} \in [0,\pi)$. Note that, in this model, $\delta_{\rm SPF}$  quantifies the magnitude of the characterized qubit state preparation flaws, while $\epsilon$ can incorporate other imperfections that are difficult to characterize, such as leakage of bit and basis information through additional degrees of freedom (see \cite{curras-lorenzoSecurityFramework2025} for more information).
        
        Our result can extend the phase-error-rate bound in Ref.~\cite[Eq.~(S116)]{curras-lorenzoSecurityFramework2025} to incorporate also detection efficiency mismatches, obtaining a security proof that is robust to both source and detector imperfections. As in Ref.~\cite[Section VI]{tupkaryPhaseError2025}, we consider a canonical model for Bob's detectors and assume that their detector efficiencies and dark count rates are not known precisely, but are only characterized to some known tolerances $\etachar$ and $\dcchar$, respectively, i.e.,
    \begin{equation}\label{eq:etadcmodel}
			\begin{aligned}
				\eta_{\beta_b} &\in [\etadet(1-\etachar ), \etadet(1+\etachar )], \\
				d_{\beta_b} &\in [\dcprob(1-\dcchar ), \dcprob(1+\dcchar )],
			\end{aligned}
		\end{equation} 
            where $\eta_{\beta_b}$ ($d_{\beta_b}$) is the detection efficiency (dark count rate) of the detector associated to basis $\beta$ and bit value $b$. For this particular model, the parameters $\delta_1$ and $\delta_2$ can be upper bounded, respectively, as \cite{tupkaryPhaseError2025}
            \begin{align}
            \delta_1 \leq \max \Bigg\{& \Bigg(1- \frac{1-(1-d_{\rm min})^2}{1-(1-d_{\rm max})^2}\Bigg) \frac{d_{\rm max}(2-d_{\rm min})}{1-(1-d_{\rm min})^2}, \nonumber \\
            & 4 \Bigg|1 - \sqrt{1-(1-d_{\rm min})^2(1-r_{\eta})}\Bigg|\Bigg\},
            \end{align}
            and
            \begin{align}
            \delta_2 \leq \max &\Bigg\{ 1 - \frac{1-(1-d_{\rm min})^2}{1-(1-d_{\rm max})^2}, (1-d_{\rm min})^2(1-r_\eta) \Bigg\},
            \end{align}
            where, 
            \begin{align}
            &d_{\rm max} = \max_{\beta\in \{X,Z\}} \{d_{\beta_0},d_{\beta_1}\} \leq d_{\rm det} (1+\Delta_{\rm dc}), \nonumber \\
            &d_{\rm min} = \min_{\beta\in \{X,Z\}} \{d_{\beta_0},d_{\beta_1}\} \geq d_{\rm det} (1-\Delta_{\rm dc}), \nonumber \\
            & r_\eta = (1-\Delta_\eta)/(1+\Delta_\eta)
.
            \end{align}
            Note that the canonical detector model considered here is just an example. Our result is valid for all (independent across rounds) detector models as long as one can find bounds on $\delta_1$ and $\delta_2$ (which are defined in \cref{eq:deltaonedefined,eq:deltatwodefined}, see also Ref.~\cite{tupkaryPhaseError2025} for more details). Finally, we remark that our analysis covers certain models in which the detector efficiency depends on the mode of the incoming signals, and in which the eavesdropper exploits this to induce a mismatch by sending signals in specific modes, as explained in Ref.~\cite[Section VIII]{tupkaryPhaseError2025}.

            In \cref{fig:graph}, we simulate the secret-key rate of our enhanced security proof in the presence of both source and detector imperfections. 
            For this, we use the channel model in Ref.~\cite[Section SE3.A]{curras-lorenzoSecurityFramework2025} and assume the same experimental parameters as in the simulations in that work, i.e., fiber loss coefficient $\alpha_{DB} = 0.2$dB/km, dark count probability $p_d = 10^{-8}$ for all detectors \cite{pittaluga600kmRepeaterlike2021}, detector efficiency $\eta_{\rm det} = 0.73$ for all detectors, \cite{pittaluga600kmRepeaterlike2021}, error correction inefficiency $f = 1.16$. Also, we consider a number of transmitted signals $N = 10^{12}$, set $\epsilon_{\rm secr} = \epsilon_{\rm corr} = 10^{-10}$ and optimize over Alice's and Bob's basis choice probabilities. 
            To evaluate the achievable secret-key rate under source and detector imperfections, we consider $\epsilon \in \{0,10^{-6},10^{-3}\}$, $\delta_{\rm SPF} \in \{0,0.063\}$~\blue{\cite{xuExperimentalQuantum2015}} and $\Delta := \Delta_{\eta} = \Delta_{\rm dc} \in \{0,0.05\}$.

            In \cref{fig:graph}, the solid black line shows the secret-key rate in the ideal scenario, and the colored (dashed) lines show the impact of source (detector) imperfections. Note that the difference between the dashed and solid lines for the same color is roughly the same for all three cases, which indicates that considering detector imperfections has a similar impact on performance independently of the source imperfections for our proof technique. Also, when comparing with the ideal case, the penalty on the secret-key rate appears to be the combined penalty of the two analyses. \blue{In \cref{app:additional_simulations}, we provide additional simulation results investigating the effect of each parameter in isolation.}

            \blue{The secret-key rate is especially sensitive to the parameter $\epsilon$, which comes from the security proof in Ref.~\cite{curras-lorenzoSecurityFramework2025} and which evaluates the fidelity between Alice's emitted states and the qubit reference states, see \cref{eq:fidelity_constraint}. Moreover, the impact of $\epsilon$ grows rapidly as the channel losses increase. This is because, when $\epsilon > 0$, Alice's states could be linearly independent, and thus vulnerable to powerful attacks based on unambiguous state discrimination (USD) measurements. Since Eve can disguise inconclusive USD attacks as losses, the impact of the imperfection increases rapidly as the observed losses increase, see Ref.~\cite{curras-lorenzoSecurityFramework2025} for an in-depth discussion on this point. We remark that, as verified by the asymptotic numerical security analyses in Refs.~\cite{curras-lorenzoNumericalSecurity2025,pereiraOptimalKey2025}, the security proof in  Ref.~\cite{curras-lorenzoSecurityFramework2025}  is essentially (asymptotically) \textit{optimal} given its assumptions. That is, no security proof operating under the same source model could achieve substantially better key rates. Thus, to obtain a good performance, one necessarily requires either very precise source characterization (i.e., a very low value of $\epsilon$), or to consider different characterization models beyond a simple fidelity bound to a reference state. We remark that our lifting theorem could be applied to any future security proof based on phase-error estimation that takes the latter approach.}

\begin{figure}[h]	\includegraphics[width=8.5cm]{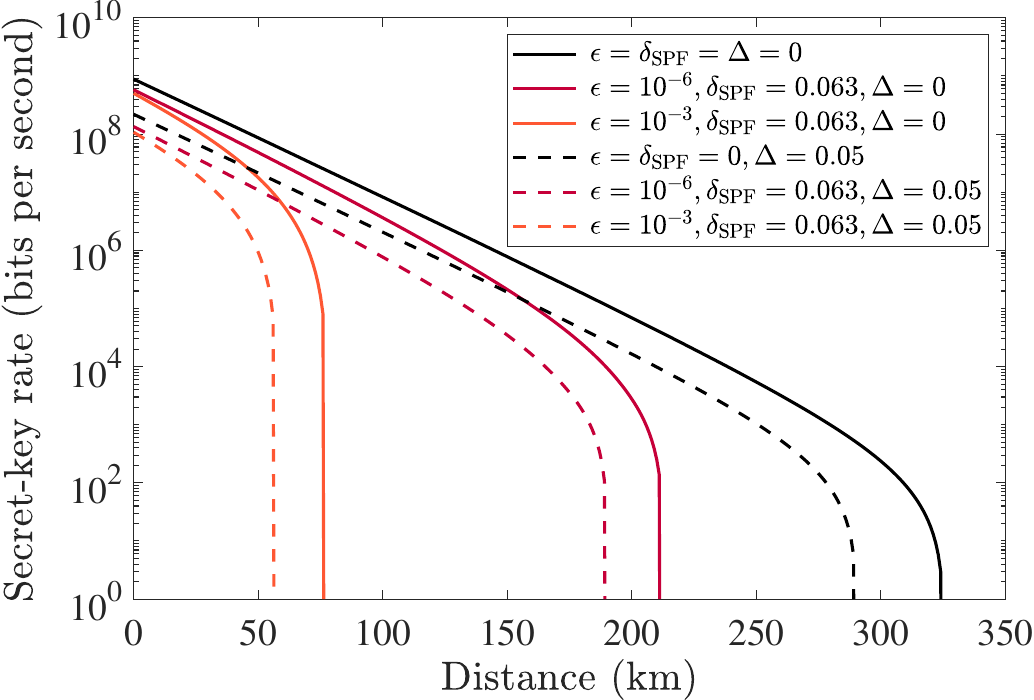} 
	\caption{Finite-size secret-key rate for a BB84-type protocol with $N=10^{12}$ transmitted signals under various imperfection scenarios. The black solid line shows the ideal BB84 case with no imperfections, and the black dashed line shows detector imperfections only (with $\Delta$ = 0.05); both of these use the phase-error bound in \cite{tupkaryPhaseError2025}. The colored solid lines show the impact of source imperfections only (for varying $\delta_{\rm SPF}$ and $\epsilon$) using the phase-error bound from Ref.~\cite{curras-lorenzoSecurityFramework2025}. The colored dashed lines correspond to the combined phase-error bound incorporating both source and detector imperfections, using the approach introduced in this work.}
    \label{fig:graph}
\end{figure}

\section*{Discussion}

We have presented a \blue{general and modular lifting theorem that} systematically extends phase-error-estimation-based security proofs for BB84-type protocols from the idealized case of basis-independent detection efficiency to realistic scenarios with detection efficiency mismatches. By applying our result to existing security analyses that already incorporate source imperfections, such as Ref.~\cite{curras-lorenzoSecurityFramework2025}, one obtains a unified finite-key security proof against general attacks that addresses both types of imperfections simultaneously---a critical requirement for practical QKD implementations.

\blue{\textbf{Comparison with prior work.} To our knowledge, only three prior works have attempted to prove security in the presence of both source and detector imperfections: Refs.~\cite{maroy2010security,sunSecurityQuantum2021,marcominiLosstolerantQuantum2025}. All three are valid only in the asymptotic regime, i.e., under the assumption that the protocol is executed for infinitely many rounds. Since real-world QKD implementations necessarily operate with finite data, such asymptotic analyses cannot rigorously certify their security. Beyond this fundamental limitation, each approach faces additional technical constraints. Refs.~\cite{maroy2010security,sunSecurityQuantum2021} employ the GLLP quantum coin approach~\cite{gottesmanSecurityQuantum2004,loSecurityQuantum2007,koashiSimpleSecurity2009} to incorporate source imperfections. This is known to be extremely pessimistic in the presence of qubit encoding flaws---imperfections in the encoding angles of Alice's states that increase their basis dependence, but do not increase their dimensionality and thus do not enable Eve to perform powerful attacks based on unambiguous state discrimination. Modern approaches to incorporate source imperfections such as Ref.~\cite{curras-lorenzoSecurityFramework2025} exploit this observation to obtain substantially better key rates. Ref.~\cite{marcominiLosstolerantQuantum2025} avoids this pessimistic performance but requires the assumption that the states emitted by Alice lie in a qubit space, and consequently cannot address imperfections that increase the dimensionality of Alice's signals, such as side-channel leakage through mode dependencies or Trojan-horse attacks. Also, it requires the unrealistic assumption that the pulses received by Bob also lie in a qubit space.

Compared to the results above, our approach is substantially more general. Rather than proving security for a particular scenario, we provide a modular lifting theorem that can be combined with more advanced phase-error-estimation-bounds for source imperfections that overcome the limitations above. In particular, by applying it to the phase-error bound in Ref.~\cite{curras-lorenzoSecurityFramework2025}, we obtain a security proof that handles both qubit encoding flaws and high-dimensional side-channel leakage, offers good performance in the presence of qubit encoding flaws, and provides rigorous finite-key security against general coherent attacks. The modularity of our approach also ensures compatibility with future advances: for instance, recent work \cite{curras-lorenzoRigorousPhaseerrorestimation2026} has shown that phase-error rate proofs addressing partially-characterized imperfect sources (such as Ref.~\cite{curras-lorenzoSecurityFramework2025}) can be systematically extended to incorporate encoding correlations, and this result can be combined directly with our lifting theorem.} 

\blue{\textbf{Alternative approaches.} An alternative to addressing detector side channels in security proofs is to eliminate them entirely through interference-based protocols such as measurement-device-independent (MDI) QKD \cite{loMeasurementDeviceIndependentQuantum2012}, twin-field QKD \cite{lucamariniOvercomingRate2018}, and asynchronous MDI-QKD \cite{xieBreakingRateLoss2022,zengModepairingQuantum2022}. These protocols are inherently immune to all detector imperfections, requiring security proofs to address only source imperfections, which previous works have investigated~\cite{navarretePracticalQuantum2021,curras-lorenzoSecurityFramework2025,guExperimentalMeasurementdeviceindependent2022}. However, these impose significantly more demanding experimental requirements, including indistinguishable photon sources, precise timing synchronization and, in some cases, phase stabilization over the communication channel. As prepare-and-measure BB84-type protocols remain the dominant choice in deployed systems, our work addresses a practically important gap by enabling rigorous security certification for such protocols.}

\blue{\textbf{Scope and limitations.} Our result applies to prepare-and-measure BB84-type protocols, where Alice prepares and emits pulses through the quantum channel. How to prove security for entanglement-based protocols with detection efficiency mismatches---where both Alice's and Bob's measurement devices may suffer from the mismatch---remains an open problem that, to our knowledge, no prior work has addressed.

As currently formulated, our lifting theorem requires that} Bob makes an active basis choice, and is therefore not directly applicable to measurement setups with passive basis choice. \blue{However, subsequent work \cite[Sec.~VIII.A]{wangPhaseError2025} has demonstrated that our lifting theorem can be extended to passive basis choice as well. \blue{Essentially, this is because, in the rounds in which Bob receives a single-photon, a passive basis choice is equivalent to an active basis choice. Using cross-basis detections and dark count rates, one can estimate Bob's outcomes for such single-photon signals, and apply our result to these outcomes.}}

Our analysis also assumes that all classical announcements occur after the full quantum communication step is complete; extending our results to allow for real-time announcements remains an open challenge for future work.

\blue{\textbf{Decoy-state protocols.}} For concreteness, in the main text, we have focused on single-photon protocols. However, as explained in Methods below, \blue{our lifting theorem extends directly to decoy-state \cite{loDecoyState2005,maPracticalDecoy2005,hwangQuantumKey2003} protocols: given a phase-error-based security proof for a decoy-state BB84-type protocol with source imperfections but ideal detectors, our result immediately extends it to incorporate detector efficiency mismatches. This extension requires only updating the single-photon phase-error rate bound in an analogous manner to the single-photon-source case. Consequently, one can expect a similar performance decrease when transitioning from ideal to imperfect detectors in decoy-state implementations as observed in our single-photon simulations.

We note, however, that unlike in the single-photon case, no phase-error-based finite-key security proof against general attacks currently exists for decoy-state BB84 with \textit{general} source imperfections---that is, one capable of simultaneously incorporating the various imperfections present in realistic decoy-state transmitters. Existing proofs address isolated imperfections, and many are restricted to the asymptotic regime and/or limited eavesdropping attacks. For example: Ref.~\cite{tamakiDecoystateQuantum2016} addresses leakage of intensity setting information, but only asymptotically; Refs.~\cite{zapateroSecurityQuantum2021,sixtoSecurityDecoyState2022} treat intensity correlations, but only asymptotically and for restricted attacks; Ref.~\cite{curras-lorenzoSecurityQuantum2023} considers phase correlations between pulses, but only asymptotically and against collective attacks; Refs.~\cite{caoDiscretephaserandomizedCoherent2015,sixtoSecretKey2023} analyze discrete and imperfect phase randomization, but only asymptotically and against collective attacks; and Ref.~\cite{sixtoQuantumKey2025} addresses information leakage from various transmitter components, but only asymptotically.

The development of finite-key security proofs against general attacks for decoy-state protocols with general and realistic source imperfections remains an active area of research. The value of our lifting theorem is that once such a proof becomes available, it can be \textit{immediately} extended to incorporate detector imperfections, without any additional theoretical work. This modularity ensures the long-term relevance of our result.}

\section*{Author Contributions}

{D.T.\ and S.N.\ initiated this project. All authors contributed to the theoretical results. M.P.\ and G.C.-L.\ were responsible for performing the simulations used to generate \cref{fig:graph}, and for the proofs in \cref{app:monotonicity}. G.C.-L.\ and M.P.\ wrote most of the paper, with contributions from all authors.}

\section*{Acknowledgments}

{We thank Kiyoshi Tamaki, Norbert L\"utkenhaus, Marcos Curty and Álvaro Navarrete for valuable discussions.  We acknowledge support from the Galician Regional Government (consolidation of research units: atlanTTic), the Spanish Ministry of Economy and Competitiveness (MINECO), the Fondo Europeo de Desarrollo Regional (FEDER) through the grant No.~PID2024-162270OB-I00, MICIN with funding from the European Union NextGenerationEU (PRTRC17.I1) and the Galician Regional Government with own funding through the “Planes Complementarios de I+D+I con las Comunidades Autonomas” in Quantum Communication, the “Hub Nacional de Excelencia en Comunicaciones Cuanticas” funded by the Spanish Ministry for Digital Transformation and the Public Service and the European Union NextGenerationEU, the European Union’s Horizon Europe Framework Programme under the Marie Sklodowska-Curie Grant No.~101072637 (Project QSI), the project “Quantum Secure Networks Partnership” (QSNP, grant agreement No 101114043) and the European Union via the European Health and Digital Executive Agency (HADEA) under the Project QuTechSpace (grant 101135225) and the Project IberianQCI (grant 101249593), as well as the Programa de Cooperación Interreg VI-A España--Portugal (POCTEP) 2021--2027 through the project QUANTUM\_IBER\_IA. G.C.-L. acknowledges funding from the European Union's Horizon Europe research and innovation programme under the Marie Skłodowska-Curie Postdoctoral Fellowship grant agreement No.\ 101149523. Part of this work was performed at the Institute for Quantum
Computing, at the University of Waterloo, which is supported by Innovation, Science, and Economic Development
Canada. This work was supported by NSERC under the Discovery Grants Program, Grant No. 341495. D.T.\ is supported by the Mike and Ophelia Laziridis Fellowship.}

\section*{Methods}

\subsection*{Application to decoy-state protocols}

Our result can also be applied to extend phase-error-estimation based security proofs for decoy-state protocols to cover scenarios with detection efficiency mismatches. Ideally, in a decoy-state BB84 protocol, in each round, Alice selects a bit-and-basis setting $j \in \{0,1,2,3\}$ with probability $p_j$ and an intensity setting $\mu$ with probability $p_\mu$, and then generates a state
\begin{equation}
\label{eq:decoy_ideal_form}
    \rho_{j,\mu} = \sum_{m} p_{m\vert \mu} \ketbra{\psi_{m,j}}_{a},
\end{equation}
where $p_{m\vert \mu}$ follows a Poisson distribution of mean $\mu$ and $\ket{\psi_{m,j}}_a$ is an $m$-photon state with perfectly encoded bit-and-basis information $j$. Ref.~\cite{tupkaryPhaseError2025} shows how to extend the standard decoy-state proofs assuming such ideal sources to cover detection efficiency mismatches.
In practice, however, due to source imperfections, the emitted states will not be exactly in the form of \cref{eq:decoy_ideal_form}. This may not only be due to imperfections introduced by the bit-and-basis encoder, but also imperfections in the implementation of the decoy-state method itself. In particular, the standard decoy-state analysis assumes that (1) Alice can perfectly tune the intensity of her pulses, (2) Alice's pulses have a fully random phase, and (3) no information about the intensity or phase of the pulses is leaked to the outside. However, in practice, these assumptions may not be met due to many reasons: discrete phase randomization \cite{caoDiscretephaserandomizedCoherent2015,sixtoSecretKey2023}, leakage of intensity setting information \cite{tamakiDecoystateQuantum2016}, imperfect intensity tuning \cite{kaminRenyiSecurity2025}, imperfect phase randomization \cite{naharImperfectPhase2023} due to e.g.\ injection locking attacks \cite{wiesemannEvaluationQuantum2025}, leakage of phase information \cite{sixtoQuantumKey2025}, correlations between the intensity of consecutive pulses \cite{yoshinoQuantumKey2018,zapateroSecurityQuantum2021,sixtoSecretKey2023}, and correlations between the phase of consecutive pulses \cite{curras-lorenzoSecurityQuantum2023,marcominiCharacterisingHigherorder2025}. Here, we consider a very general description of the global state that would be emitted in the presence of these kind of imperfections and show that, provided that one already has a phase-error-based security proof that can incorporate them, one can apply our result to extend the proof to cover detector efficiency mismatches as well.

In a generalized decoy-state BB84-type protocol, Alice chooses a sequence of bit-and-basis setting choices $j_1^N \coloneqq j_1 ... j_N$ and intensity setting choices $\mu_1^N \coloneqq \mu_1 ... \mu_N$ with some joint probability $p_{j_1^N\!,\mu_1^N}$, and generates a global $N$-round state
\begin{equation}
    \rho_{j_1^N\!, \mu_1^N} = \sum_{m_1^N} p_{m_1^N\! \vert j_1^N\!,\mu_1^N} \ketbra*{\psi_{m_1^N\!,j_1^N\!,\mu_1^N}}_{a_1^N},
\end{equation}
where $a_1^N \coloneqq a_1 ... a_N$ is the sequence of photonic systems emitted by Alice, and $m_1^N = m_1 ... m_N$ identifies an eigenstate of the global state emitted by Alice, and can be regarded as the sequence of photon numbers (or, more generally, quasi-photon numbers) of the states emitted by Alice \footnote{Even more generally, one can consider any decomposition of the form
\begin{equation}
    \rho_{j_1^N\!, \mu_1^N} = \sum_{m_1^N} p_{m_1^N\! \vert j_1^N\!,\mu_1^N} \tau_{m_1^N\!,j_1^N\!,\mu_1^N}.
\end{equation}
where the states $\tau_{m_1^N\!,j_1^N\!,\mu_1^N}$ are possibly mixed. In this case, one can simply consider purifications $\ket*{\psi_{m_1^N\!,j_1^N\!,\mu_1^N}}_{a_1^N S_1^N}$, where the shield systems $S_1^N$ are not available to Eve.}. Note that our model takes into account that the global state emitted by Alice may be arbitrarily correlated between rounds. 
As usual, to prove security, we consider an equivalent scenario in which Alice generates a source-replacement state,
\begin{equation}
\begin{aligned}
    \ket{\Psi} =&\sum_{j_1^N\!,\mu_1^N} \sqrt{p_{j_1^N\!,\mu_1^N}} \ket*{j_1^N}_{A_1^N\!} \ket*{\mu_1^N}_{I_1^N\!} \\ &~\sum_{m_1^N} \sqrt{p_{m_1^N \vert j_1^N\!,\mu_1^N}} \ket*{m_1^N}_{M_1^N\!} \ket*{\psi_{m_1^N,j_1^N\!,\mu_1^N}}_{a_1^N\!},
\end{aligned}
\end{equation}
where $A_1^N \coloneqq A_1 ... A_N$, $I_1^N \coloneqq I_1 ... I_N$ and $M_1^N = M_1 ... M_N$ are sequences of ancillary registers containing the information about the values of $j_1^N$, $\mu_1^N$ and $m_1^N$, respectively, and the states $\{\ket*{j_1^N}_{A_1^N}\}$, $\{\ket*{\mu_1^N}_{I_1^N}\}$, and $\{\ket*{m_1^N}_{M_1^N}\}$ form orthonormal bases in their respective Hilbert spaces. Note that this is a valid source replacement scheme since
\begin{equation}
    \Tr_{A_1^N I_1^N M_1^N}\!\big[\ketbra*{\Psi}\big] = \sum_{j_1^N\!,\mu_1^N} p_{j_1^N\!,\mu_1^N}  \,\rho_{j_1^N\!, \mu_1^N} 
\end{equation}
In this scenario, Alice can measure the registers $A_1^N\!$ and $I_1^N\!$ to learn her choices of bit, basis and intensity values. Using this information, together with Bob's detection and basis announcements, Alice can determine which rounds are used for key generation. For concreteness, we assume  here that Alice and Bob extract the raw key from the events in which Bob chooses the $\zpovm$ POVM, Alice's encoding choice is $j\in \{0,1\}$ and, optionally, the rounds in which Alice selected a particular intensity choice.

As in the main text, to prove security, we consider a scenario in which Alice and Bob first learn only which rounds are used for key generation, but do not learn their bit values yet. We also assume that Alice measures the registers $M_1^N$ to learn the quasi-photon number of her signals, and uses this information to tag the key generation rounds according to their value of $m$. Again, the security of the key generated in the actual protocol can be related to the \textit{phase} errors that Alice and Bob would obtain if, in the key generation rounds, Alice measured her ancillas $A_1^N$ in the $\{\ket{+}_A,\ket{-}_A\}$ basis and Bob performed his $\xpovm$ POVM. We can define $\bm{n_{K,1}}$ as the RV associated to the number of key generation rounds such that $m=1$, and $\bm{n_{\mathrm{ph},1}}$ as the number of phase errors within these rounds. In general, a security proof for decoy-state BB84-type protocols based on phase-error estimation consists of a proof for two statements. The first statement is a lower bound on the number of key generation rounds in which $m=1$,%
\begin{equation}
\label{eq:guarantee_nK1_decoy_end_matter}
    \Prindep \left[\bm{n_{K,1}} < \mathcal{M} (\bm{\vec n_\zpovm}) \right]\leq \epssp^2,
\end{equation}
while the second statement bounds the phase-error rate within these rounds,
\begin{equation}
\label{eq2:guarantee_BIDE_decoy_end_matter}
\Prindep(\bm{e_{\mathrm{ph},1}} >    \Eindep^\mathrm{decoy}(\bm{\vec n_\xpovm},\bm{n_{K,1}})) \leq \epssrc^2,
\end{equation} 
Here, $\bm{\vec n_\zpovm}$ ($\bm{\vec n_\xpovm}$) is the random vector containing all the announced statistics from rounds in which Bob chose the $\zpovm$ ($\xpovm$) POVM, $\mathcal{M}(\bm{\vec n_\zpovm})$ is a function that provides a lower bound on $\bm{n_{K,1}}$ based on the observed statistics from the $\zpovm$-basis measurements, $\epssp^2$ is the failure probability associated with this bound,  and $\bm{e_{\mathrm{ph},1}} = \bm{n_{\mathrm{ph},1}}/\bm{n_{K,1}}$ represents the phase-error rate for the $n=1$ component. 
Assuming for simplicity the typical scenario in which  $\Eindep^\mathrm{decoy}$ is non-increasing with $\bm{n_{K,1}}$, \cref{eq:guarantee_nK1_decoy_end_matter,eq2:guarantee_BIDE_decoy_end_matter} can be combined using the union bound,
\begin{align}
    &\Prindep \left[\bm{n_{K,1}} < \mathcal{M} (\bm{\vec n_\zpovm}) \cup \bm{e_{\mathrm{ph},1}} >     \Eindep^\mathrm{decoy}\big(\bm{\vec n_\xpovm},\mathcal{M} (\bm{\vec n_\zpovm})\big) \right] \nonumber \\
    &\leq \epssp^2 + \epssrc^2,
    \label{eq:decoy_union_bound_ideal}
\end{align}
and using this combined bound, one can directly determine the final secret key length and its security parameter, see \cref{eq:general_union_bound_decoy,eq:key_length_decoy_endmatter} below.

Our result can be directly applied to extend any bounds of the form in \cref{eq:guarantee_nK1_decoy_end_matter,eq2:guarantee_BIDE_decoy_end_matter} to cover detector efficiency mismatches. For simplicity, let us assume here the typical case in which $\Eindep^\mathrm{decoy}$ is non-increasing with $\bm{n_{K,1}}$ and $\Findepdecoy (\bm\vecnX,\bm{n_{K,1}}) \coloneqq \bm{n_{K,1}} \Eindepdecoy (\bm\vecnX,\bm{n_{K,1}})$ is non-decreasing with $\bm{n_{K,1}}$. Then, we have that, when Bob's measurement setup suffers from a detection efficiency mismatch parametrized by $\delta_1$ and $\delta_2$, it holds that
{\small \begin{equation} \label{eq:decoy_union_bound_mismatch}
\begin{aligned}
    &\Prdep \Bigg(\bm{n_{K,1}} < \mathcal{M} (\bm{\vec n_\zpovm}) \, \cup \bm{e_{\mathrm{ph},1}} > \\ 
    &\frac{\Eindepdecoy \left(\bm\vecnX,\mathcal{M} (\bm{\vec n_\zpovm})\right) + \delta_1 + \gamma^{\epsATb}_{\mathrm{bin}}(\mathcal{M} (\bm{\vec n_\zpovm}), \deltaone)}{1-\delta_2-\gamma^{\epsATc}_{\mathrm{bin}}\big(\mathcal{M} (\bm{\vec n_\zpovm}), \delta_2\big)}\Bigg) \\
    &\leq \epssp^2+  \epssrc^2 + \epsATb^2 + \epsATc^2.
\end{aligned}
\end{equation}}
\!For a proof of these statements, as well as a more general result that does not need $\Eindepdecoy$ and $\Findepdecoy$ to satisfy any monotonicity conditions, see \cref{app:decoy_state_scenarios}.

Given a general bound of the form in \cref{eq:decoy_union_bound_ideal} or \cref{eq:decoy_union_bound_mismatch}, i.e.,
\begin{align}
\label{eq:general_union_bound_decoy}
    \Prboth \left[\bm{n_{K,1}} < \mathcal{M} (\bm{\vec n_\zpovm}) \cup \bm{e_{\mathrm{ph},1}} >   \Eboth'\big(\bm{\vec n_\xpovm},\bm{\vec n_\zpovm}\big) \right] \leq \epspe^2,
\end{align}
security can be directly established using either EUR and LHL \cite{tomamichelUncertaintyRelation2011,tomamichelTightFinitekey2012,tomamichelLargelySelfcontained2017} or phase-error correction arguments \cite{koashiSimpleSecurity2009}. As demonstrated in \cite[Theorem 3]{tupkaryPhaseError2025}, a bound of this form directly yields finite-key security against general attacks in the variable-length framework \cite{ben-orUniversalComposable2005}, where the key length depends on observations during protocol execution. Specifically, one can achieve a secure key length 
\begin{align}
    \bm{l} &= \mathcal{M} (\bm{\vec n_\zpovm}) \Big[1-h\big(\Eboth'(\bm{\vec n_\xpovm},\bm{\vec n_\zpovm})\big) \Big] - \bm\lambdaec \nonumber \\
    &- 2 \log(1/2 \epspa) - \log(2/\epsev),  \label{eq:key_length_decoy_endmatter}
\end{align}
that is $(\epscorr+\epssec)$-secure, with $\epscorr = \epsev$ and $\epssec = 2 \epspe + \epspa$.

%% file: supp-body.tex
\section{Description of Bob's measurement in the phase-error estimation protocol}
\label{app:bobs_measurement}

Here, we show how to find an equivalence between Bob's measurement in the actual protocol and a scenario in which he first applies filters to determine whether or not he obtains a detection, followed by a two-outcome POVM that is guaranteed to output a bit value. We then use this equivalence to define Bob's action in the phase-error estimation protocol, for both the basis-independent and basis-dependent cases. We remark that the content of this Supplementary Note is simply a reformulation of the measurement process as an equivalent multi-step measurement process that implements the same quantum to classical map, which is itself an old idea. Our approach is most similar to that presented in \cite[Lemma 1 and Section IV]{tupkaryPhaseError2025}, and is included here for completeness.

\subsection{Basis-independent detection efficiency}
\label{subapp:BIDE}

In this case, Bob performs two POVMs $\mathcal{\zpovm} = \{\Gamma_0^{(\zpovm)},\Gamma_1^{(\zpovm)},\Gamma_\bot^{(\zpovm)}\}$ and $\mathcal{\xpovm} = \{\Gamma_0^{(\xpovm)},\Gamma_1^{(\xpovm)},\Gamma_\bot^{(\xpovm)}\}$ whose operators corresponding to an undetected round are the same for both bases, i.e., $\Gamma_\bot^{(\zpovm)} = \Gamma_\bot^{(\xpovm)}$. This implies that Bob could have substituted his actual measurement by the following:

\begin{enumerate}
    \item For each round, apply the filtering POVM $\{\tilde F, \mathbb{I} - \tilde F\}$, where $\tilde F = \Gamma_0^{(\zpovm)} + \Gamma_1^{(\zpovm)} = \Gamma_0^{(\xpovm)} + \Gamma_1^{(\xpovm)}$.
    \item For the rounds that pass the filter, decide his basis choice $\beta \in \{\zpovm,\xpovm\}$.
    \item Measure these states using the two-outcome POVM $\{G_0^{(\beta)},G_1^{(\beta)}\}$, where
    \begin{equation}
        G_b^{(\beta)} := \sqrt{\tilde F^+} \, \Gamma_b^{(\beta)} \, \sqrt{\tilde F^+} + P_b^{(\beta)},
    \end{equation}
    with $\tilde F^+$ denoting the pseudoinverse of $\tilde F$, and $P_b^{(\beta)}$ any positive operators satisfying $\sum_{b \in \{0,1\}} P_b^{(\beta)} = \mathbb{I} - \Pi_{\tilde F}$, where $\Pi_{\tilde F}$ denotes the projector onto the support of $\tilde F$.
\end{enumerate}

As explained in the main text, thanks to this equivalence, we can assume that Alice and Bob learn which subset of $\mathcal{\zpovm}$ rounds are used to generate the key before learning their bit values. In the phase-error estimation protocol, we consider that, in these rounds, Bob performs the two-outcome POVM $\{G_0^{(\xpovm)},G_1^{(\xpovm)}\}$ rather than $\{G_0^{(\zpovm)},G_1^{(\zpovm)}\}$. We refer to this scenario as $\Sindep$.

\subsection{Basis-dependent detection efficiency}
\label{subapp:BDDE}

In this case, the operators corresponding to an undetected round are in general not equal, i.e., $\Gamma_\bot^{(\zpovm)} \neq \Gamma_\bot^{(\xpovm)}$. Here, we can define the basis dependent filters $\tilde F_\zpovm = \Gamma_0^{(\zpovm)} + \Gamma_1^{(\zpovm)}$ and $\tilde F_\xpovm = \Gamma_0^{(\xpovm)} + \Gamma_1^{(\xpovm)}$ and assume the following equivalent scenario for Bob:

\begin{enumerate}
    \item For each round, decide his basis choice $\beta \in \{\zpovm,\xpovm\}$.
    \item For each round, apply the filter $\tilde F_\beta$ corresponding to the selected basis $\beta$.
    \item For the rounds that pass the filter, measure using the two-outcome POVM $\{G_0^{(\beta)},G_1^{(\beta)}\}$, where
    \begin{equation}
        G_b^{(\beta)} := \sqrt{\tilde F_\beta^+} \, \Gamma_b^{(\beta)} \, \sqrt{\tilde F_\beta^+} + P_b^{(\beta)}, 
    \end{equation}
    with $P_b^{(\beta)}$ being any positive operators satisfying $\sum_{b \in \{0,1\}} P_b^{(\beta)} = \mathbb{I} - \Pi_{\tilde F_\beta}$.

    \end{enumerate}

    Again, thanks to this equivalence, we can assume that Alice and Bob learn which subset of $\mathcal{\zpovm}$ rounds are used to generate the key before learning their bit values. In the phase-error estimation protocol, we consider that, in these rounds, Bob performs the two-outcome POVM $\{G_0^{(\xpovm)},G_1^{(\xpovm)}\}$ rather than $\{G_0^{(\zpovm)},G_1^{(\zpovm)}\}$.

    \subsubsection*{Definition of \texorpdfstring{$\delta_1$ and $\delta_2$}{delta1 and delta2}}
    
    Now, we use the following idea from Ref.~\cite{tupkaryPhaseError2025}: if we find an operator $\tilde F$ such that $\tilde F \geq \tilde F_\zpovm$ and $\tilde F \geq \tilde F_\xpovm$, we can define the following equivalent steps for Bob, which replace steps 1 and 2 above:
    \begin{enumerate}
\setcounter{enumi}{-1}
        \item For each round, apply the basis-independent filtering POVM $\{\tilde F, \mathbb{I}-\tilde F\}$.
        \item For the rounds that pass the filter, decide his basis choice $\beta \in \{\zpovm,\xpovm\}$.
               \item Perform the second basis-dependent filtering POVM $\{F_\beta,\mathbb{I}-F_\beta\}$, where
        \begin{equation}
            F_\beta := \sqrt{\tilde F^+} \, \tilde F_\beta \, \sqrt{\tilde F^+} + \mathbb{I} - \Pi_{\tilde F}. 
        \end{equation}
        
    \end{enumerate}

Note that this is equivalent since applying first the filter $\tilde F$ and then the filter $F_\zpovm$ ($F_\xpovm$) is equivalent to applying the overall filter $\tilde F_\zpovm$ ($\tilde F_\xpovm$). This allows us to define the following parameters from \cite{tupkaryPhaseError2025} which are used in the proof of our main result in \cref{app:proof_main_result}:
\begin{equation} \label{eq:deltaonedefined}
\begin{aligned}
    \delta_1 &= \Big\lVert\big(\mathbb{I}_A \otimes\sqrt{F_\zpovm}\big) G_{\neq}^{(\xpovm)} \big(\mathbb{I}_A \otimes\sqrt{F_\zpovm}\big) \\
    &- \big(\mathbb{I}_A \otimes\sqrt{F_\xpovm}\big) G_{\neq}^{(\xpovm)} \big(\mathbb{I}_A \otimes\sqrt{F_\xpovm}\big) \Big \rVert_\infty,
\end{aligned}
\end{equation}
with $G_{\neq}^{(\xpovm)} = \ketbra*{+}_A \otimes G_1^{(\xpovm)} + \ketbra*{-}_A \otimes G_0^{(\xpovm)}$ and 
\begin{equation}\label{eq:deltatwodefined}
    \delta_2 = \big\lVert \mathbb{I} - F_\zpovm \big\rVert_\infty.
\end{equation}

\section{Proof of our main result}
\label{app:proof_main_result}

Here, we prove our main result in \cref*{eq2:objective} of the main text. For this, we define a single ``global'' phase-error estimation protocol, which is depicted in \cref{fig:diagram}, where Bob makes an active decision whether to run the basis-independent detection efficiency scenario ($\Sindep$) or the basis-dependent detection efficiency scenario  ($\Sdep$):

\begin{enumerate}[start=1]
    \item Alice prepares her global source-replacement state in \cref*{eq:source_replacement_state} of the main text.

    \item Bob applies the basis independent filter $\{\tilde F,  \mathbb{I} - \tilde F\}$. Let $\mathcal{\tilde N}$ be the  set of rounds that pass this filter.
    
    \item For each round in $\mathcal{\tilde N}$, Bob chooses the $\zpovm$ or $\xpovm$ POVM, but does not perform any measurement yet. Let $\mathcal{\tilde N}_\zpovm$ ($\mathcal{\tilde N}_{\xpovm}$) be the subset of $\mathcal{\tilde N}$ for which Bob chooses the $\zpovm$ ($\xpovm$) POVM.
    
    \item For each round in $\mathcal{\tilde N}_\zpovm$, Alice attempts a projection onto the subspace spanned by $\{\ket{0}_A,\ket{1}_A\}$. Let $\mathcal{\tilde N}_K$ be the subset of rounds for which this projection is successful, and let $\bm{\tilde n_K}$ be the random variable associated to the size of $\mathcal{\tilde N}_K$.
    
    \item \label{listitem:test_rounds_meas} For each round in $\mathcal{\tilde N}_{\xpovm}$, Bob applies the $\xpovm$-basis-dependent filter $\{F_\xpovm,\mathbb{I}-F_\xpovm\}$. Then, for the rounds that  pass the filter, Alice performs a projection onto $\{\ket{j}_A\}_{j\in \{0,1,...\}}$, learning her setting choice, while Bob performs the two-outcome POVM $\{G_{0}^{(\xpovm)},G_1^{(\xpovm)}\}$, learning his bit value. Let $\bm{\vecnX}$ be a random vector containing all the RVs associated to the number of (fine-grained) outcomes of these rounds.

    \item \label{listitem:key_rounds_meas} Bob decides whether to run basis independent ($\Sindep$) or the basis dependent ($\Sdep$) scenario. Then, for each round in $\mathcal{\tilde N}_K$:

    \begin{enumerate}
        \item[]If $\Sindep$: Bob applies the $\xpovm$-dependent filter $\{F_\xpovm,\mathbb{I}-F_\xpovm\}$; 
        \item[] If $\Sdep$: Bob applies the $\zpovm$-dependent filter $\{F_\zpovm,\mathbb{I}-F_\zpovm\}$. 
    \end{enumerate}

    Let $\mathcal{N}_K$ be the subset of rounds in $\mathcal{\tilde N}_K$ that pass the filter, and let $\bm{n_K}$ be the random variable associated to the size of $\mathcal{N}_K$.

    \item \label{listitem:final_ph_rounds_meas} Let $\{G_{\neq}^{(X)},G_{=}^{(X)}\}$ be Alice and Bob's joint two-outcome $X$-basis POVM, i.e., $G_{\neq}^{(X)} = \ketbra*{+}_A \otimes G_1^{(\xpovm)} + \ketbra*{-}_A \otimes G_0^{(\xpovm)}$. For each round in $\mathcal{N}_K$, Alice and Bob measure $\{G_{\neq}^{(X)},G_{=}^{(X)}\}$. Let $\bm{n_\mathbf{ph}}$ be the RV associated to the number of events in which they obtain the outcome $G_{\neq}^{(X)}$, and let $\bm{e_\mathbf{ph}} \coloneqq \bm{n_\mathbf{ph}}/\bm{n_K}$.\\
\end{enumerate}

   Note that, in $\Sindep$, Bob applies exactly the same filters to the key rounds and the test rounds. Therefore, this is equivalent to the phase-error estimation protocol of a scenario with basis-independent detection efficiency, defined in \cref{subapp:BIDE}. On the other hand, $\Sdep$ corresponds to the phase-error estimation protocol of a scenario with basis-dependent detection efficiency parameterized by $\delta_1$ and $\delta_2$, see \cref{subapp:BDDE}. We define $\Prindep$ ($\Prdep$) as the probability measure conditional on the selection of $\Sindep$ ($\Sdep$).

    In what follows, we prove the results presented in the main text. In \cref{fig:diagram}, we provide a  representation of the scenario described above and the proof below.

\begin{widetext}

\begin{figure}[ht]
  \centering
\resizebox{\linewidth}{!}{%
\begin{tikzpicture}[node distance=2cm] 
   			\node (start) {$ \rho_{A^nB^nE^n}$ };
   			\node (filtering) [process, below of=start] {Basis-independent filtering \\
   				using $\{ \tilde{F} , \mathbb{I} - \tilde{F}\}$};
   			
   			\node (center) [process, below of=filtering] {Alice and Bob's Basis Choice. \\ Test vs Key. };

   			\node(keystate) [point, right of=center,xshift = 2cm,yshift = -2cm] {};
   			
   			\node(keystateabove) [emptyprocess, above of=keystate,yshift = -1.3cm] {$\rho_{A^{\bm{\tilde n_K}}B^{\bm{\tilde n_K}} C^n E^n}$};

   			\node (test) [process, below left of=center, xshift=-2.5cm, yshift=-1.5cm] 
   			{ \textbf{Testing Rounds} \\
   		 Bob uses $\xpovm$ POVM. \\
            $\bm 
            \vecnX$ is observed.};

   			
   			\node (keyvirtXX) [virtualprocess, below right of=keystate, xshift=-4.5cm, yshift=-1cm] 
   			{  \textbf{$\Sindep$ scenario} \\
   			 Bob applies the $\xpovm$-basis-dependent filter \\
                $\{F_\xpovm,\mathbb{I}-F_\xpovm\}$ 	};

   			\node (keyXX) [virtualprocess, below of=keyvirtXX, yshift= -2cm] 
   			{  Measure $\bm{n_K}$ surviving rounds\\
               using POVM $\{G_{\neq}^{(\xpovm)},G_{=}^{(\xpovm)}\}$. \\
               Phase error rate   $\bm{e_\mathbf{ph}} \coloneqq \bm{n_\mathbf{ph}}/\bm{n_K}$.};
   			
   	\node (keyvirt) [virtualprocess, below right of=keystate, xshift=2.5cm, yshift=-1cm] 
   				{  \textbf{$\Sdep$ scenario} \\
   			 Bob applies the $\zpovm$-basis-dependent filter \\
                $\{F_\zpovm,\mathbb{I}-F_\zpovm\}$ 	};
   			
   			\node (key) [virtualprocess, below  of=keyvirt,yshift = -2cm] 
   			{ Measure $\bm{n_K}$  surviving rounds\\
               using POVM $\{G_{\neq}^{(\xpovm)},G_{=}^{(\xpovm)}\}$. \\
               Phase error rate   $\bm{e_\mathbf{ph}} \coloneqq \bm{n_\mathbf{ph}}/\bm{n_K}$. };

   			\draw [arrow] (start) -- ++ (filtering); 
   			\draw [arrow, ] (filtering) -- ++ (center) node[midway,right] {$\tilde{F}$};
   			\draw [arrow] (center.south) -- (keystate) ;
   			\draw [dashedarrow] (keystate) -- (keyvirtXX);
   			\draw [dashedarrow] (keystate) -- (keyvirt);
   			\draw [arrow](center.south) -- (test);

               \draw [dashedarrow] (keyvirt) -- (key) node [midway, left] {$F_\zpovm$};
   			
   			\draw [arrow] (filtering.east) -- ++(2,0) node[midway, above] {$\mathbb{I}-\tilde{F}$} node[right] {Discard};
   			
   			\draw [arrow] (center.east) -- ++(1.8,0) node[midway, above] {} node[right] {Discard};

               \draw [arrow,dashed] (keyvirtXX) -- ++(-5,0) node[midway, above] {$\mathbb{I}-{F_\xpovm}$} node[left] {Discard};

               \draw [arrow,dashed] (keyvirt) -- ++(5,0) node[midway, above] {$\mathbb{I}-{F_\zpovm}$} node[right] {Discard};
   			
   			\draw [dashedarrow] (keyvirtXX)  -- ++ (keyXX) node[midway, right] {$F_\xpovm$};
   			
   			
   			\draw ($(keyXX.south west)!0.2!(keyXX.north west)$) edge[dashed, <-, thick, bend left=50, color = blue]  node[pos=0.2,below,align=center,yshift=-0.7em,xshift=-3em] {Estimate $\bm{\nph}$ from $\bm\vecnX$ \\ using existing bound (e.g., \cite{curras-lorenzoSecurityFramework2025})} ($(test.south west)!0.3!(test.north west)$);

   			\draw (keyvirtXX) edge [dashed, <-, thick, bend right, color = blue, align = center]  node[midway, left] { $~\bm{n_K} \leq \bm{\tilde n_K}$ 
               } (keyXX);

   			\draw (keyvirt) edge [dashed, ->, thick, bend left, color = blue, align = center]  node[midway, right] {~$\bm{n_K} \gtrsim \bm{\tilde n_K} (1-\delta_2)$ \\ \cite[Lemma 4]{tupkaryPhaseError2025}} (key);
   
   			\draw (keyvirtXX) edge [dashed, ->, thick, bend right, color = blue,align = center,yshift=-1.5em]  node[midway,below] {$ \bm{\nph}\ \lesssim \bm{\nph} + \bm{\tilde n_K} \delta_1$ \\ \cite[Lemma 3]{tupkaryPhaseError2025}} (keyvirt);

   		\end{tikzpicture}}
			
			\caption{Flowchart describing the scenario considered in \cref{app:proof_main_result} (black lines and text) and an informal representation of the proof in \cref{thm:maintheorem} (blue lines and text). We consider a global phase-error estimation protocol in which Bob makes an active decision to run $\Sindep$ (the basis-independent scenario) or $\Sdep$ (the basis-dependent scenario). For clarity, we include the scenario $\Sindep,\Sdep$ as superscripts when denoting random variables that depend on those scenarios in the figure. We use \cite[Lemma 4]{tupkaryPhaseError2025} to relate the statistics of $\bm{n_{K}}$ and $\bm{\tilde n_{K}}$ conditional on $\Sdep$. We use \cite[Lemma 3]{tupkaryPhaseError2025} to relate the conditional distribution of $\bm\nph$ depending on whether Bob chooses $\Sindep$ or $\Sdep$.}
			\label{fig:diagram}
		\end{figure}

    \begin{theorem} \label{thm:maintheorem}
Consider the scenario defined above, and suppose that
\begin{equation}
\label{eqapp:guarantee_BIDE}
\Prindep(\bm{e_\mathbf{ph}} > 	\Eindep(\bm{\vec n_\xpovm},\bm{n_K})) \leq \epssrc^2.
\end{equation}   
Let $\mathcal{N}_{\leq}(m) \coloneqq \{0,1,...,m\}$ be the set of non-negative integers until $m$. Then, for any $\epsATb,\epsATc>0$, define
\begin{equation}
    \mathcal{W}_{\delta_2}(m) \coloneqq \mathcal{N}_{\leq} \left(\left\lfloor\frac{m}{1-\delta_2-\gamma^{\epsATc}_{\mathrm{bin}}(m, \delta_2)} \right \rfloor\right).
\end{equation}
Then, 
\begin{equation}
    \Prdep \left(\bm\eph > \frac{\max_{n \in \mathcal{W}_{\delta_2}(\bm{n_K})}  n \, \Eindep (\bm\vecnX,n)}{\bm{n_K}}  + \frac{ \delta_1 + \gamma^{\epsATb}_{\mathrm{bin}}(\bm{n_K}, \deltaone)}{1-\delta_2-\gamma^{\epsATc}_{\mathrm{bin}}(\bm{n_K}, \delta_2)} \right) \leq \epssrc^2 + \epsATb^2 + \epsATc^2,
\end{equation}
where $\delta_1, \delta_2$ are defined in \cref{eq:deltaonedefined,eq:deltatwodefined}, and $\gamma_\mathrm{bin}$ is the finite-size deviation term defined in \cref{eq:gammabin_defined}.

\end{theorem}

\begin{proof}
It is helpful to refer to \cref{fig:diagram} during this proof. Our starting bound, \cref{eqapp:guarantee_BIDE}, can be equivalently expressed as
\begin{align}
\label{eq2:guarantee_nph}
&\Prindep(\bm{n_\mathbf{ph}} > 	\Findep(\bm{\vecnX},\bm{n_K})) 
= \Prindep(\bm{e_\mathbf{ph}}\, \bm{n_{K}} > 	\Findep(\bm{\vecnX},\bm{n_K})) \nonumber  \\ 
&= \Prindep(\bm{e_\mathbf{ph}} > 	\Findep(\bm{\vecnX},\bm{n_K})/\bm{n_K}) = \Prindep(\bm{e_\mathbf{ph}} > 	\Eindep(\bm{\vec n_\xpovm},\bm{n_K})) \underset{{\cref{eqapp:guarantee_BIDE}}}{\leq} \epssrc^2,
\end{align}
where $\Findep\big(\bm{\vec n_\xpovm}, \bm{n_{K}}\big) \coloneqq \bm{n_K} \Eindep\big(\bm{\vec n_\xpovm}, \bm{n_{K}}\big),$ as defined in \cref*{eq:Findepdef} of the main text.
We wish to transform the above bound to the scenario $\Sdep$. To do so, we first get rid of the $\bm{n_K}$ dependence via

\begin{equation}
\begin{gathered}
\label{eq2:guarantee_our_proof_modified}
\Prindep(\bm{n_\mathbf{ph}} > \max_{n \in \mathcal{N}_{\leq}(\bm{\tilde n_K)}}	\Findep(\bm{\vecnX},n)) \leq \Prindep(\bm{n_\mathbf{ph}} > 	\Findep(\bm{\vecnX},\bm{n_K})) \underset{{\cref{eq2:guarantee_nph}}}{\leq} \epssrc^2,
\end{gathered}
\end{equation}
which holds since $\bm{n_K} \leq \bm{\tilde n_K}$ and thus $\bm{n_K} \in \mathcal{N}_{\leq}(\bm{\tilde n_K})$.

Now, consider the state of the rounds in $\mathcal{\tilde N}_K$ just before Step (\ref*{listitem:key_rounds_meas}), after observing the event $\Omega(\tilde n_K,\vecnX)$, corresponding to  $\bm{\tilde n_K} = \tilde n_K$ and $\bm{\vecnX} = \vecnX$. Alice and Bob's measurements in Steps (\ref*{listitem:key_rounds_meas}) and (\ref*{listitem:final_ph_rounds_meas}) can be regarded as a single three-outcome POVM $\{\sqrt{F^{(\beta)}} G_{\neq}^{(\xpovm)} \sqrt{F^{(\beta)}},\sqrt{F^{(\beta)}} G_{=}^{(\xpovm)} \sqrt{F^{(\beta)}},\mathbb{I}-F^{(\beta)}\}$ applied to each round in $\mathcal{\tilde N}_K$, where $\beta = \xpovm$ for $\Sindep$ and $\beta = \zpovm$ for $\Sdep$. The POVM elements corresponding to a phase error in each scenario, i.e., $\sqrt{F^{(\xpovm)}} G_{\neq}^{(\xpovm)} \sqrt{F^{(\xpovm)}}$ and $\sqrt{F^{(\zpovm)}} G_{\neq}^{(\xpovm)} \sqrt{F^{(\zpovm)}}$, are close to one another, as quantified by $\deltaone$ (see \cref{eq:deltaonedefined}). Therefore, the statistics of $\bm{\nph}$ in the two scenarios should also be similar. This can be formalized by applying \cite[Lemma 3]{tupkaryPhaseError2025}, obtaining 
the following statement: for any $e \geq 0$,
\begin{equation}  \label{eq2:lemma3_start}
				\Prdep(\bm{n_\mathbf{ph}} > \tilde n_K \big( e +  \deltaone + \gamma^{\epsATb}_{\mathrm{bin}}(\tilde n_K,\deltaone) \big) )_{ | \Omega(\tilde n_K,\vecnX)}\leq \Prindep(\bm{n_\mathbf{ph}} > \tilde n_K e)_{ | \Omega(\tilde n_K,\vecnX)}  +\epsATb^2,
\end{equation}
where $\gamma^{\varepsilon}_{\mathrm{bin}}(x,y)$ is a function introducing a finite-size deviation term arising from sampling the binomial distribution, defined in \cref{eq:gammabin_defined}.  Substituting $e
= \max_{n \in \mathcal{N}_{\leq}(\tilde n_K)} \Findep(\vecnX, n)/\tilde n_K$, we obtain

\begin{align}
&\Prdep\Big(\bm{n_\mathbf{ph}} > \max_{n \in \mathcal{N}_{\leq}(\tilde n_K)} \Findep(\vecnX, n) + \tilde n_K \left( \deltaone + \gamma^{\epsATb}_{\mathrm{bin}}(\tilde n_K,\deltaone) \right) \Big)_{ | \Omega(\tilde n_K,\vecnX)} \nonumber \\
&\leq \Prindep\left(\bm{n_\mathbf{ph}} > \max_{n \in \mathcal{N}_{\leq}(\tilde n_K)} \Findep(\vecnX, n) \right)_{ | \Omega(\tilde n_K,\vecnX)} + \epsATb^2.
\label{eq2:lemma3_step2}
\end{align}

This statement can be generalized to hold for the random variables themselves via:
\begin{equation}  \label{eq2:lemma3_step4}
\begin{aligned}
&\Prdep\left(\bm{n_\mathbf{ph}} > \max_{n \in \mathcal{N}_{\leq}(\bm{\tilde n_K})} \Findep(\bm{\vecnX},n) +\bm{\tilde n_K} \big(  \deltaone + \gamma^{\epsATb}_{\mathrm{bin}}(\bm{n_K},\deltaone) \big) \right)\\
&\underset{\substack{\bm{n_K} \leq \bm{\tilde n_K} \\ \gamma_\mathrm{bin}^{\varepsilon} \mathrm{ decreas.}}}{\leq}\Prdep\left(\bm{n_\mathbf{ph}} > \max_{n \in \mathcal{N}_{\leq}(\bm{\tilde n_K})} \Findep(\bm{\vecnX},n) +\bm{\tilde n_K} \big(  \deltaone + \gamma^{\epsATb}_{\mathrm{bin}}(\bm{\tilde n_K},\deltaone) \big) \right)\\
&= \sum_{\tilde n_K,\vecnX} \Prdep(\bm{\tilde n_K} = \tilde n_K,\bm{\vecnX} = \vecnX) \Prdep\left(\bm{n_\mathbf{ph}} > \max_{n \in \mathcal{N}_{\leq}(\tilde n_K)} \Findep(\vecnX,n) +\tilde n_K \big(  \deltaone + \gamma^{\epsATb}_{\mathrm{bin}}(\tilde n_K,\deltaone) \big) \right)_{ | \Omega(\tilde n_K,\vecnX)} \\
&\underset{(*)}{=}   \sum_{\tilde n_K,\vecnX} \Prindep(\bm{\tilde n_K} = \tilde n_K,\bm{\vecnX} = \vecnX)\Prdep\left(\bm{n_\mathbf{ph}} > \max_{n \in \mathcal{N}_{\leq}(\tilde n_K)} \Findep(\vecnX,n) +\tilde n_K \big(  \deltaone + \gamma^{\epsATb}_{\mathrm{bin}}(\tilde n_K,\deltaone) \big) \right)_{ | \Omega(\tilde n_K,\vecnX)} \\
&\underset{\cref{eq2:lemma3_step2}}{\leq} \sum_{\tilde n_K,\vecnX} \Prindep(\bm{\tilde n_K} = \tilde n_K,\bm{\vecnX} = \vecnX) \bigg[ \Prindep\left(\bm{n_\mathbf{ph}} > \max_{n \in \mathcal{N}_{\leq}(\tilde n_K)} \Findep(\vecnX,n)\right)_{ | \Omega(\tilde n_K,\vecnX)} +  \epsATb^2 \bigg]  \\
&= \bigg[\sum_{\tilde n_K,\vecnX} \Prindep(\bm{\tilde n_K} = \tilde n_K,\bm{\vecnX} = \vecnX)\Prindep\left(\bm{n_\mathbf{ph}} > \max_{n \in \mathcal{N}_{\leq}(\tilde n_K)} \Findep(\vecnX,n)\right)_{ | \Omega(\tilde n_K,\vecnX)} \bigg] +  \epsATb^2 \\
&= \Prindep\left(\bm{n_\mathbf{ph}} > \max_{n \in \mathcal{N}_{\leq}(\bm{\tilde n_K})} \Findep(\bm{\vecnX},n)\right)  +\epsATb^2 \\
&\underset{\cref{eq2:guarantee_our_proof_modified}}{\leq} \epssrc^2 + \epsATb^2.
\end{aligned}
\end{equation}

In the equality marked by an asterisk, we have used the fact that
\begin{equation}
    \Prindep(\bm{\tilde n_K} = \tilde n_K,\bm{\vecnX} = \vecnX) = \Prdep(\bm{\tilde n_K} = \tilde n_K,\bm{\vecnX} = \vecnX),
\end{equation}
since the value of the random variables $\bm{\tilde n_K}$ and $\bm{\vecnX}$ is fixed by the end of Step (\ref*{listitem:test_rounds_meas}), and therefore cannot be affected by Bob's decision of whether to run $\Sindep$ or $\Sdep$ in Step (\ref*{listitem:key_rounds_meas}).

The bound from \cref{eq2:lemma3_step4} depends on the value of $\bm{\tilde n_K}$, which is not observed in the protocol. We will now get rid of this dependence by using the fact that $\bm{n_K}$ is obtained by discarding a small number of rounds (due to basis-\textit{dependent} loss) from $\bm{\tilde n_K}$, and thus the two should be close. The rate at which this discarding occurs is upper bounded by $\deltatwo$ defined in \cref{eq:deltatwodefined}. Thus, applying \cite[Lemma 4]{tupkaryPhaseError2025} to the state of the rounds in $\mathcal{\tilde N}_{K}$ conditional on choosing scenario $\Sdep$ and on the event $\bm{\tilde n_K} = \tilde n_K$, we obtain the statement
\begin{equation}
\label{eq2:lemma4_1}
    \Prdep\Big(\bm{n_K} < \tilde n_K \big({1-\delta_2-\gamma^{\epsATc}_{\mathrm{bin}}(\tilde n_K, \delta_2)} \big)\Big)_{| \Omega(\tilde n_K)} \leq \epsATc^2.
\end{equation}
Using the fact that $\gamma^{\varepsilon}_{\mathrm{bin}}(x,y)$ is non-increasing with respect to $x$ and the fact that $\Prdep(\bm{n_K} >  \tilde n_K)_{| \Omega(\tilde n_K)} = 0$, we replace the  $\tilde{n}_K$ with $\bm{n_K}$ and obtain
\begin{equation}
\label{eq2:lemma4_2}
\begin{aligned}
    &\Prdep\Big(\bm{n_K} < \bm{\tilde n_K} \big({1-\delta_2-\gamma^{\epsATc}_{\mathrm{bin}}(\bm{n_K}, \delta_2)} \big)\Big) \\
    &= \sum_{\tilde n_K} \Prdep(\bm{\tilde n_K} = \tilde n_K) \Prdep\Big(\bm{n_K} < \tilde n_K \big({1-\delta_2-\gamma^{\epsATc}_{\mathrm{bin}}(\bm{n_K}, \delta_2)} \big)\Big)_{| \Omega(\tilde n_K)} \\ 
    &\leq \sum_{\tilde n_K} \Prdep(\bm{\tilde n_K} = \tilde n_K) \Prdep\Big(\bm{n_K} < \tilde n_K \big({1-\delta_2-\gamma^{\epsATc}_{\mathrm{bin}}(\tilde n_K, \delta_2)} \big)\Big)_{| \Omega(\tilde n_K)} \\
    &\underset{{\cref{eq2:lemma4_1}}}{\leq} \sum_{\tilde n_K} \Prdep(\bm{\tilde n_K} = \tilde n_K) \, \epsATc^2 \\
    &= \epsATc^2.
\end{aligned}
\end{equation}
Using \cref{eq2:lemma4_2}, the fact that $\bm{\tilde n_K} \geq \bm{n_K}$, and the fact that $\bm{n_K}$ and $\bm{\tilde n_K}$ can only take integer values, we obtain
\begin{equation}
\label{eq2:lemma4_3}
\begin{aligned}
    &\Prdep\Big(\bm{\tilde n_K} \notin \mathcal{V}_{\delta_2}(\bm{n_K}) \Big) \leq  \epsATc^2,
\end{aligned}
\end{equation}
where \begin{equation}
\mathcal{V}_{\delta_2} (m) \coloneqq \left\{m,m+1, \ldots, \left \lfloor \frac{m}{1-\delta_2-\gamma^{\epsATc}_{\mathrm{bin}}(m, \delta_2)} \right \rfloor\right\}.
\end{equation}

Using this, we obtain the following bound which can be computed from the observed statistics
\begin{equation}
\label{eq:almost_end_proof}
\begin{aligned}
&\Prdep \left(\bm\nph > \max_{\tilde n \in \mathcal{V}_{\delta_2}(\bm{n_K})} \left[\max_{n \in \mathcal{N}_{\leq}(\tilde n)}  \Findep (\bm\vecnX,n)  + \tilde n (\delta_1 + \gamma^{\epsATb}_{\mathrm{bin}}(\bm{n_K}, \deltaone))\right] \right) \\
&\leq \Prdep \left(\bm\nph > \max_{\tilde n \in \mathcal{V}_{\delta_2}(\bm{n_K})} \left[\max_{n \in \mathcal{N}_{\leq}(\tilde n)}  \Findep (\bm\vecnX,n)  + \tilde n (\delta_1 + \gamma^{\epsATb}_{\mathrm{bin}}(\bm{n_K}, \deltaone))\right] \cap \bm{\tilde n_K} \in \mathcal{V}_{\delta_2}(\bm{n_K})  \right)
+ \Prdep \left(\bm{\tilde n_K} \notin \mathcal{V}_{\delta_2}(\bm{n_K}) \right) \\
&\underset{{\cref{eq2:lemma4_3}}}{\leq}  \Prdep \left(\bm\nph > \max_{\tilde n \in \mathcal{V}_{\delta_2}(\bm{n_K})} \left[\max_{n \in \mathcal{N}_{\leq}(\tilde n)}  \Findep (\bm\vecnX,n)  + \tilde n (\delta_1 + \gamma^{\epsATb}_{\mathrm{bin}}(\bm{n_K}, \deltaone))\right] \cap \bm{\tilde n_K} \in \mathcal{V}_{\delta_2}(\bm{n_K})  \right)
+ \epsATc^2
\\
&\leq \Prdep \left(\bm\nph > \max_{n \in \mathcal{N}_{\leq}(\bm{\tilde n_K})}  \Findep (\bm\vecnX,n)  + \bm{\tilde n_K} (\delta_1 + \gamma^{\epsATb}_{\mathrm{bin}}(\bm{n_K}, \deltaone))\right)
+ \epsATc^2 \\
&\underset{{\cref{eq2:lemma3_step4}}}{\leq} \epssrc^2 + \epsATb^2 + \epsATc^2.
\end{aligned}
\end{equation}

Finally, the double optimization in the above bound can be simplified to obtain
\begin{equation}
\label{eq:end_proof}
\begin{aligned}
&\Prdep \left(\bm\nph > \max_{n \in \mathcal{W}_{\delta_2}(\bm{n_K})}  \Findep (\bm\vecnX,n)  + \frac{\bm{n_K} \left(\delta_1 + \gamma^{\epsATb}_{\mathrm{bin}}(\bm{n_K}, \deltaone)\right)}{1-\delta_2-\gamma^{\epsATc}_{\mathrm{bin}}(\bm{n_K}, \delta_2)} \right) \\
&= \Prdep \left(\bm\nph > \max_{\tilde n \in \mathcal{V}_{\delta_2}(\bm{n_K})} \max_{n \in \mathcal{N}_{\leq}(\tilde n)}  \Findep (\bm\vecnX,n)  + \frac{\bm{n_K} \left(\delta_1 + \gamma^{\epsATb}_{\mathrm{bin}}(\bm{n_K}, \deltaone)\right)}{1-\delta_2-\gamma^{\epsATc}_{\mathrm{bin}}(\bm{n_K}, \delta_2)} \right) \\
&= \Prdep \left(\bm\nph > \left[\max_{\tilde n \in \mathcal{V}_{\delta_2}(\bm{n_K})} \, \max_{n \in \mathcal{N}_{\leq}(\tilde n)}  \Findep (\bm\vecnX,n)\right]  + \left[\max_{\tilde n \in \mathcal{V}_{\delta_2}(\bm{n_K})} \tilde n (\delta_1 + \gamma^{\epsATb}_{\mathrm{bin}}(\bm{n_K}, \deltaone)) \right]\right) \\
&\leq \Prdep \left(\bm\nph > \max_{\tilde n \in \mathcal{V}_{\delta_2}(\bm{n_K})} \left[\max_{n \in \mathcal{N}_{\leq}(\tilde n)}  \Findep (\bm\vecnX,n)  + \tilde n (\delta_1 + \gamma^{\epsATb}_{\mathrm{bin}}(\bm{n_K}, \deltaone))\right] \right) \\
&\underset{{\cref{eq:almost_end_proof}}}{\leq} \epssrc^2 + \epsATb^2 + \epsATc^2,
\end{aligned}
\end{equation}
or equivalently,
\begin{equation}
\label{eq:end_proof_eph}
\begin{aligned}
&\Prdep \left(\bm\eph > \frac{\max_{n \in \mathcal{W}_{\delta_2}(\bm{n_K})}  n \, \Eindep (\bm\vecnX,n)}{\bm{n_K}}  + \frac{ \delta_1 + \gamma^{\epsATb}_{\mathrm{bin}}(\bm{n_K}, \deltaone)}{1-\delta_2-\gamma^{\epsATc}_{\mathrm{bin}}(\bm{n_K}, \delta_2)} \right) \\
&= \Prdep \left(\bm\eph \bm{n_K} > \max_{n \in \mathcal{W}_{\delta_2}(\bm{n_K})}  \Findep (\bm\vecnX,n)  + \frac{\bm{n_K} \left(\delta_1 + \gamma^{\epsATb}_{\mathrm{bin}}(\bm{n_K}, \deltaone)\right)}{1-\delta_2-\gamma^{\epsATc}_{\mathrm{bin}}(\bm{n_K}, \delta_2)}\right) \\
&= \Prdep \left(\bm\nph > \max_{n \in \mathcal{W}_{\delta_2}(\bm{n_K})}  \Findep (\bm\vecnX,n)  + \frac{\bm{n_K} \left(\delta_1 + \gamma^{\epsATb}_{\mathrm{bin}}(\bm{n_K}, \deltaone)\right)}{1-\delta_2-\gamma^{\epsATc}_{\mathrm{bin}}(\bm{n_K}, \delta_2)} \right) \\
&\underset{\substack{\cref{eq:end_proof}}}{\leq} \epssrc^2 + \epsATb^2 + \epsATc^2,
\end{aligned}
\end{equation}
as we wanted to prove.
\end{proof}

\begin{corollary} \label{maincorollary}
Assume that \cref{eqapp:guarantee_BIDE} holds, and $\Findep(\bm{\vec n_\xpovm},\bm{n_K}) \coloneqq \bm{n_K} \Eindep(\bm{\vec n_\xpovm},\bm{n_K})$ is non-decreasing with respect to $\bm{n_{K}}$. Then, for any $\epsATb,\epsATc>0$,
\begin{equation}\label{eq:manCorTighter}
\begin{aligned}
    \Prdep \left(\bm\eph > \frac{\mathcal{E}_{0,0}\left(\bm{\vec n_\xpovm}, \left \lfloor\frac{\bm{n_K}}{1-\delta_2-\gamma^{\epsATc}_{\mathrm{bin}}(\bm{n_K}, \delta_2)}\right \rfloor \right) + \deltaone + \gamma^{\epsATb}_{\mathrm{bin}}(\bm{n_{K}}, \deltaone)}{1-\delta_2-\gamma^{\epsATc}_{\mathrm{bin}}(\bm{n_{K}}, \delta_2)} \right) \leq \epssrc^2 + \epsATb^2 + \epsATc^2.
\end{aligned}
\end{equation}
If additionally, $ 	\Eindep(\bm{\vec n_\xpovm},\bm{n_K})$ is non-increasing with respect to $\bm{n_{K}}$, then
\begin{equation} \label{eq:mainCorLooser}
    \Prdep \left(\bm\eph > \frac{\mathcal{E}_{0,0}\big(\bm{\vec n_\xpovm}, \bm{n_{K}}\big) + \deltaone + \gamma^{\epsATb}_{\mathrm{bin}}(\bm{n_{K}}, \deltaone)}{1-\delta_2-\gamma^{\epsATc}_{\mathrm{bin}}(\bm{n_{K}}, \delta_2)} \right) \leq \epssrc^2 + \epsATb^2 + \epsATc^2.
\end{equation}
\end{corollary}

\begin{proof}

By assumption, $\Findep(\bm{\vec n_\xpovm},\bm{n_K})$ is non-decreasing with respect to $\bm{n_{K}}$, so
\begin{equation}
\begin{gathered}
\label{eq5:corollary_0.5}
    \max_{n \in \mathcal{W}_{\delta_2} (\bm{n_K})} n \, \Eindep (\bm\vecnX,n) =  \max_{n \in \mathcal{W}_{\delta_2}(\bm{n_K})}  \Findep (\bm\vecnX,n) =   \Findep \left(\bm\vecnX,\max_{n \in \mathcal{W}_{\delta_2}(\bm{n_K})} n\right) \\=   \Findep \left(\bm\vecnX,  \left \lfloor\frac{\bm{n_K}}{1-\delta_2-\gamma^{\epsATc}_{\mathrm{bin}}(\bm{n_K}, \delta_2)}\right \rfloor\right) = \Findep (\bm\vecnX,\bm{n^\star}) = \bm{n^\star} \Eindep (\bm\vecnX,\bm{n^\star}) ,
\end{gathered}
\end{equation}
where we have defined
\begin{equation}
\label{eq:n_star_def}
 \bm{n^\star} \coloneqq  \left \lfloor\frac{\bm{n_K}}{1-\delta_2-\gamma^{\epsATc}_{\mathrm{bin}}(\bm{n_K}, \delta_2)}\right \rfloor.   
\end{equation}

Therefore,
\begin{equation}
\label{eq5:corollary_1}
\begin{aligned}
&\Prdep\left(\bm{\eph} > \frac{\mathcal{E}_{0,0}\big(\bm{\vec n_\xpovm}, \bm{n^\star}\big) + \deltaone + \gamma^{\epsATb}_{\mathrm{bin}}(\bm{n_K}, \deltaone)}{1-\delta_2-\gamma^{\epsATc}_{\mathrm{bin}}(\bm{n_K}, \delta_2)}\right) \\
& \underset{\substack{\cref{eq:n_star_def}}}{\leq} \Prdep \left(\bm\eph >  \frac{\bm{n^\star} \Eindep (\bm\vecnX,\bm{n^\star})}{\bm{n_K}}  + \frac{\delta_1 + \gamma^{\epsATb}_{\mathrm{bin}}(\bm{n_K}, \deltaone)}{1-\delta_2-\gamma^{\epsATc}_{\mathrm{bin}}(\bm{n_K}, \delta_2)} \right) \\
&\underset{\substack{\cref{eq5:corollary_0.5}}}{\leq}    \Prdep \left(\bm\eph > \frac{\max_{n \in \mathcal{W}_{\delta_2}(\bm{n_K})}  \Findep (\bm\vecnX,n)}{\bm{n_K}}  + \frac{\delta_1 + \gamma^{\epsATb}_{\mathrm{bin}}(\bm{n_K}, \deltaone)}{1-\delta_2-\gamma^{\epsATc}_{\mathrm{bin}}(\bm{n_K}, \delta_2)} \right) \\
    &\underset{\substack{\cref{eq:end_proof_eph}}}{\leq} \epssrc^2 + \epsATb^2 + \epsATc^2.
\end{aligned}   
\end{equation}

By assumption, $\Eindep(\bm{\vec n_\xpovm},\bm{n_K})$ is non-increasing with respect to $\bm{n_{K}}$, so
\begin{equation}
\label{eq5:corollary_1.5}
     \Eindep (\bm\vecnX,\bm{n^{\star}}) \leq  \Eindep (\bm\vecnX,\bm{n_K})
\end{equation}
where we have used $\bm{n^\star} \geq \bm{n_K}$. Therefore,
\begin{equation}
\label{eq5:final_bound_Nph}
\begin{aligned}
&\Prdep\left(\bm{\eph} > \frac{\mathcal{E}_{0,0}\big(\bm{\vec n_\xpovm}, \bm{n_K}\big) + \deltaone + \gamma^{\epsATb}_{\mathrm{bin}}(\bm{n_K}, \deltaone)}{1-\delta_2-\gamma^{\epsATc}_{\mathrm{bin}}(\bm{n_K}, \delta_2)}\right) \\
& \underset{{\cref{eq5:corollary_1.5}}}{\leq}  \Prdep\left(\bm{\eph} > \frac{\mathcal{E}_{0,0}\big(\bm{\vec n_\xpovm}, \bm{n^\star}\big) + \deltaone + \gamma^{\epsATb}_{\mathrm{bin}}(\bm{n_K}, \deltaone)}{1-\delta_2-\gamma^{\epsATc}_{\mathrm{bin}}(\bm{n_K}, \delta_2)}\right) \\
&\underset{\substack{\cref{eq5:corollary_1}}}{\leq} \epssrc^2 + \epsATb^2 + \epsATc^2,
\end{aligned}
\end{equation}
as we wanted to prove.
\end{proof}
\end{widetext}

\begin{remark}
Note that, in the main text, $\Prindep$ and $\Prdep$ are simply defined as the probability measure of the phase-error estimation protocol of a scenario satisfying basis-independent detection efficiency ($\Sindep$) and a scenario with a detection efficiency mismatch parameterized by $\delta_1$ and $\delta_2$ ($\Sdep$), respectively. However, here, as a tool to prove our claim, we have embedded both scenarios into a single global scenario in which Bob makes an active decision whether to run $\Sindep$ or $\Sdep$. This global scenario is described earlier in \cref{app:proof_main_result}. Therefore, in the context of our proof, $\Prindep$ and $\Prdep$ are defined as conditional probability measures conditional on Bob's choice. Of course, every bound derived for these conditional measures holds verbatim for the scenarios considered in the main text. This is the reason why our result can be used as a tool that takes a phase-error bound derived for a scenario with basis-independent detection efficiency and extends it to a scenario with a detection efficiency mismatch, as explained in the main text.
\end{remark}

\begin{remark}
    If the source is assumed to emit ideal BB84 states, then one can easily finds a bound on the phase-error rate in the basis-independent detection efficiency scenario by applying Serfling's inequality. In particular, as shown in \cite{tupkaryPhaseError2025}, \cref{eqapp:guarantee_BIDE} holds if one sets
    \begin{equation}
    \label{eq:Eindep_perfect_source}
        \Eindep\big((\bm{e_X},\bm{n_X}),\bm{n_K}\big) = \bm{e_X} + \gamma_{\mathrm{serf}}(\bm{n_X},\bm{n_K}),
    \end{equation}
    with
    \begin{equation}
        \gamma_{\mathrm{serf}}(\bm{n_X},\bm{n_K}) = \sqrt{\frac{\ln(1/\epssrc^2) \cdot (\bm{n_K} + \bm{n_X})(\bm{n_X} + 1)}{\bm{n_K} \cdot \bm{n_X}^2}},
    \end{equation}
    where $\bm{n_X}$ is the total number of rounds in which Alice and Bob announce the $X$ basis, and $\bm{e_X}$ is the observed bit-error rate within these rounds. Then if one substitutes \cref{eq:Eindep_perfect_source} into \cref{eq:mainCorLooser}, one obtains exactly the same result as in \cite[Eq. (36)]{tupkaryPhaseError2025}. However, as noted in the corollary, one could substitute \cref{eq:Eindep_perfect_source} into \cref{eq:manCorTighter} instead to obtain a slightly tighter result.
\end{remark}

\section{Proof of our result for generalized decoy-state scenarios}
\label{app:decoy_state_scenarios}

As explained in Methods, our result can be applied to extend phase-error-based security proofs for generalized decoy-state scenarios to cover detection efficiency mismatches. Here, we prove the specific technical results that allow this application. 

\begin{widetext}
\begin{theorem} \label{thm:maintheoremdecoy}
Let $\Sindep$ be the phase-error-estimation protocol of a particular generalized decoy-state scenario when Bob's measurement setup satisfies the basis-independent detection efficiency condition, and let $\Sdep$ be the phase-error estimation protocol of the same generalized decoy-state scenario when Bob's measurement setup suffers from a detection efficiency mismatch parameterized by $\delta_1$ and $\delta_2$, defined in \cref{eq:deltaonedefined,eq:deltatwodefined}. Consider a global scenario like that defined in \cref{app:proof_main_result}, in which Bob makes an active decision whether to run $\Sindep$ or $\Sdep$. Let $\Prindep$ ($\Prdep$) be the probability measure conditional on the choice of $\Sindep$ ($\Sdep$). Suppose that  
\begin{equation}
\label{eq:guarantee_nK1_decoy}
   \Prindep \left[\bm{n_{K,1}} < \mathcal{M} (\bm{\vec n_\zpovm}) \right]\leq \epssp^2,
\end{equation}
and
\begin{equation}
\label{eq2:guarantee_BIDE_decoy}
\Prindep(\bm{e_{\mathrm{ph},1}} > 	\Eindep^\mathrm{decoy}(\bm{\vec n_\xpovm},\bm{n_{K,1}})) \leq \epssrc^2,
\end{equation} 
Then, for any $\epsATb$ and $\epsATc$,
\begin{equation} \label{eq:objective_combined}
\begin{aligned}
   &\Prdep \left( \bm{n_{K,1}} < \mathcal{M} (\bm{\vec n_\zpovm}) \bigcup \bm{e_{\mathrm{ph},1}} > \max_{\mathcal{M} (\bm{\vec n_\zpovm})\leq \hat{n}\leq\bm{n_{K}}}\frac{\max_{n \in \mathcal{W}_{\delta_2}(\hat{n})}  n \, \Eindepdecoy (\bm\vecnX,n)}{\hat{n}}  + \frac{ \delta_1 + \gamma^{\epsATb}_{\mathrm{bin}}(\mathcal{M} (\bm{\vec n_\zpovm}), \deltaone)}{1-\delta_2-\gamma^{\epsATc}_{\mathrm{bin}}(\mathcal{M} (\bm{\vec n_\zpovm}), \delta_2)} \right)\\
   &\leq \epssp^2 + \epssrc^2 + \epsATb^2 + \epsATc^2,
\end{aligned}
\end{equation}
where $\gamma_\mathrm{bin}$ is a finite-size deviation term defined in \cref{eq:gammabin_defined}.
\end{theorem}

\begin{proof}
   First, note that \cref{eq:guarantee_nK1_decoy} is a statement that depends only on the outcomes of the $\zpovm$ POVM rounds, and therefore it must hold even if basis-independent detection efficiency does not hold. That is, if \cref{eq:guarantee_nK1_decoy} holds, then
   \begin{equation}
   \label{eq:nK1_decoy_extension}
       \Prdep \left[\bm{n_{K,1}} < \mathcal{M} (\bm{\vec n_\zpovm}) \right]\leq \epssp^2,
   \end{equation}
   also holds.
   The proof that \cref{eq2:guarantee_BIDE_decoy} implies
       \begin{equation} \label{eq:err_decoy_extension}
           \Prdep \left(\bm{e_{\mathrm{ph},1}} > \frac{\max_{n \in \mathcal{W}_{\delta_2}(\bm{n_{K,1}})}  n \, \Eindepdecoy (\bm\vecnX,n)}{\bm{n_{K,1}}}  + \frac{ \delta_1 + \gamma^{\epsATb}_{\mathrm{bin}}(\bm{n_{K,1}}, \deltaone)}{1-\delta_2-\gamma^{\epsATc}_{\mathrm{bin}}(\bm{n_{K,1}}, \delta_2)}\right) \leq \epssrc^2 + \epsATb^2 + \epsATc^2,
       \end{equation}
   is essentially the same as the proof of \cref{thm:maintheorem}. Finally, using the union bound on \cref{eq:nK1_decoy_extension,eq:err_decoy_extension} with the fact that $\bm{n_{K,1}}\leq \bm{n_{K}}$ and that $\gamma^{\varepsilon}_{\mathrm{bin}}(n, \delta)$ is a non-increasing function of $n$ completes the proof.
\end{proof}

As in the single-photon case, we expect that, in most security proofs, the function $\Eindepdecoy (\bm\vecnX,\bm{n_{K,1}})$ satisfies certain monotonicity conditions. In this case, the result can be simplified a lot, as shown below.

\begin{corollary}
   Consider the same as in \cref{thm:maintheoremdecoy}, and assume that the function $\Findepdecoy (\bm\vecnX,\bm{n_{K,1}}) \coloneqq \bm{n_{K,1}} \Eindepdecoy (\bm\vecnX,\bm{n_{K,1}})$ is non-decreasing with respect to $\bm{n_{K,1}}$, and that $\Eindepdecoy (\bm\vecnX,\bm{n_{K,1}})$ is non-increasing with respect to $\bm{n_{K,1}}$. Then,
\begin{equation} \label{eq:objective_combined_3}
\begin{aligned}
   &\Prdep \left(\bm{n_{K,1}} < \mathcal{M} (\bm{\vec n_\zpovm}) \bigcup \bm{e_{\mathrm{ph},1}} > \frac{\Eindepdecoy \left(\bm\vecnX,\left \lfloor\frac{\mathcal{M} (\bm{\vec n_\zpovm})}{1-\delta_2-\gamma^{\epsATc}_{\mathrm{bin}}\big(\mathcal{M} (\bm{\vec n_\zpovm}), \delta_2\big)}\right \rfloor\right) + \delta_1 + \gamma^{\epsATb}_{\mathrm{bin}}(\mathcal{M} (\bm{\vec n_\zpovm}), \deltaone)}{1-\delta_2-\gamma^{\epsATc}_{\mathrm{bin}}\big(\mathcal{M} (\bm{\vec n_\zpovm}), \delta_2\big)}\right)\\
   &\leq \epssrc^2 + \epsATb^2 + \epsATc^2 + \epssp^2.
\end{aligned}
\end{equation}

Also,
\begin{equation} \label{eq:objective_combined_4}
\begin{aligned}
   &\Prdep \left(\bm{n_{K,1}} < \mathcal{M} (\bm{\vec n_\zpovm}) \bigcup \bm{e_{\mathrm{ph},1}} > \frac{\Eindepdecoy \left(\bm\vecnX,\mathcal{M} (\bm{\vec n_\zpovm})\right) + \delta_1 + \gamma^{\epsATb}_{\mathrm{bin}}(\mathcal{M} (\bm{\vec n_\zpovm}), \deltaone)}{1-\delta_2-\gamma^{\epsATc}_{\mathrm{bin}}\big(\mathcal{M} (\bm{\vec n_\zpovm}), \delta_2\big)}\right)\\
   &\leq \epssrc^2 + \epsATb^2 + \epsATc^2 + \epssp^2,
\end{aligned}
\end{equation}
which is a simpler but slightly less tight bound.
\end{corollary}

\begin{proof}
   If $\Findepdecoy (\bm\vecnX,\bm{n_{K,1}})$ is non-decreasing with respect to $\bm{n_{K,1}}$, by applying to \cref{eq:objective_combined} exactly the same derivations as in \cref{eq5:corollary_0.5,eq:n_star_def,eq5:corollary_1,eq5:corollary_1.5}, we obtain
       \begin{equation} \label{eq:objective_combined_2}
\begin{aligned}
   &\Prdep \left(\bm{n_{K,1}} < \mathcal{M} (\bm{\vec n_\zpovm}) \bigcup \bm{e_{\mathrm{ph},1}} > \max_{\mathcal{M} (\bm{\vec n_\zpovm})\leq \hat{n}\leq\bm{n_{K}}}\frac{\Eindepdecoy \left(\bm\vecnX,\left \lfloor\frac{\hat{n}}{1-\delta_2-\gamma^{\epsATc}_{\mathrm{bin}}(\hat{n}, \delta_2)}\right \rfloor\right)}{1-\delta_2-\gamma^{\epsATc}_{\mathrm{bin}}(\hat{n}, \delta_2)}  + \frac{ \delta_1 + \gamma^{\epsATb}_{\mathrm{bin}}(\mathcal{M} (\bm{\vec n_\zpovm}), \deltaone)}{1-\delta_2-\gamma^{\epsATc}_{\mathrm{bin}}(\mathcal{M} (\bm{\vec n_\zpovm}), \delta_2)}\right)\\
   &\leq \epssrc^2 + \epsATb^2 + \epsATc^2 + \epssp^2,
\end{aligned}
\end{equation}

Then, if $\Eindepdecoy (\bm\vecnX,\bm{n_{K,1}})$ is non-increasing with respect to $\bm{n_{K,1}}$, we have that
{\small       \begin{equation} \label{eq:objective_combined_2.1}
\begin{aligned}
   &\Prdep \left(\bm{n_{K,1}} < \mathcal{M} (\bm{\vec n_\zpovm}) \bigcup \bm{e_{\mathrm{ph},1}} > \frac{\Eindepdecoy \left(\bm\vecnX,\left \lfloor\frac{\mathcal{M} (\bm{\vec n_\zpovm})}{1-\delta_2-\gamma^{\epsATc}_{\mathrm{bin}}\big(\mathcal{M} (\bm{\vec n_\zpovm}), \delta_2\big)}\right \rfloor\right) + \delta_1 + \gamma^{\epsATb}_{\mathrm{bin}}(\mathcal{M} (\bm{\vec n_\zpovm}), \deltaone)}{1-\delta_2-\gamma^{\epsATc}_{\mathrm{bin}}\big(\mathcal{M} (\bm{\vec n_\zpovm}), \delta_2\big)}\right)\\
   &= \Prdep \left(\bm{n_{K,1}} < \mathcal{M} (\bm{\vec n_\zpovm}) \bigcup \bm{e_{\mathrm{ph},1}} > \frac{\Eindepdecoy \left(\bm\vecnX,\min_{\mathcal{M} (\bm{\vec n_\zpovm})\leq \hat{n}\leq\bm{n_{K}}} \left  \lfloor\frac{\hat{n}}{1-\delta_2-\gamma^{\epsATc}_{\mathrm{bin}}(\hat{n}, \delta_2)}\right \rfloor\right)}{\min_{\mathcal{M} (\bm{\vec n_\zpovm})\leq \hat{n}\leq\bm{n_{K}}} 1-\delta_2-\gamma^{\epsATc}_{\mathrm{bin}}(\hat{n}, \delta_2)}  + \frac{ \delta_1 + \gamma^{\epsATb}_{\mathrm{bin}}(\mathcal{M} (\bm{\vec n_\zpovm}), \deltaone)}{1-\delta_2-\gamma^{\epsATc}_{\mathrm{bin}}(\mathcal{M} (\bm{\vec n_\zpovm}), \delta_2)}\right)\\
   &\leq \Prdep \left(\bm{n_{K,1}} < \mathcal{M}  (\bm{\vec n_\zpovm}) \bigcup \bm{e_{\mathrm{ph},1}} > \frac{\max_{\mathcal{M} (\bm{\vec n_\zpovm})\leq \hat{n}\leq\bm{n_{K}}} \Eindepdecoy \left(\bm\vecnX,\left \lfloor\frac{\hat{n}}{1-\delta_2-\gamma^{\epsATc}_{\mathrm{bin}}(\hat{n}, \delta_2)}\right \rfloor\right)}{\min_{\mathcal{M} (\bm{\vec n_\zpovm})\leq \hat{n}\leq\bm{n_{K}}} 1-\delta_2-\gamma^{\epsATc}_{\mathrm{bin}}(\hat{n}, \delta_2)}  + \frac{ \delta_1 + \gamma^{\epsATb}_{\mathrm{bin}}(\mathcal{M} (\bm{\vec n_\zpovm}), \deltaone)}{1-\delta_2-\gamma^{\epsATc}_{\mathrm{bin}}(\mathcal{M} (\bm{\vec n_\zpovm}), \delta_2)} \right)\\
   &\leq \Prdep \left(\bm{n_{K,1}} < \mathcal{M} (\bm{\vec n_\zpovm}) \bigcup \bm{e_{\mathrm{ph},1}} > \max_{\mathcal{M} (\bm{\vec n_\zpovm})\leq \hat{n}\leq\bm{n_{K}}}\frac{\Eindepdecoy \left(\bm\vecnX,\left \lfloor\frac{\hat{n}}{1-\delta_2-\gamma^{\epsATc}_{\mathrm{bin}}(\hat{n}, \delta_2)}\right \rfloor\right)}{1-\delta_2-\gamma^{\epsATc}_{\mathrm{bin}}(\hat{n}, \delta_2)}  + \frac{ \delta_1 + \gamma^{\epsATb}_{\mathrm{bin}}(\mathcal{M} (\bm{\vec n_\zpovm}), \deltaone)}{1-\delta_2-\gamma^{\epsATc}_{\mathrm{bin}}(\mathcal{M} (\bm{\vec n_\zpovm}), \delta_2)}\right)\\
   &\underset{\substack{\cref{eq:objective_combined_2}}}{\leq} \epssrc^2 + \epsATb^2 + \epsATc^2 + \epssp^2,
\end{aligned}
\end{equation}}
\!i.e., \cref{eq:objective_combined_3}. Using the again the fact that $\Eindepdecoy (\bm\vecnX,\bm{n_{K,1}})$ is non-increasing with respect to $\bm{n_{K,1}}$, we immediately recover \cref{eq:objective_combined_4} as well.
\end{proof}

As explained in Methods, any of the bounds in \cref{eq:objective_combined},  \cref{eq:objective_combined_3} or \cref{eq:objective_combined_4} can be directly used to determine the length of the final output key and its security parameter.

\section{Proof of monotonicity of the phase-error bound in
  \texorpdfstring{\cite{curras-lorenzoSecurityFramework2025}}{Curras-Lorenzo et al.}}
\label{app:monotonicity}

As explained in the main text, Ref.~\cite{curras-lorenzoSecurityFramework2025} derives a bound on the number of phase-errors of the form
\begin{equation}
\label{eq:nph_bound}
\Pr(\bm{n_\mathbf{ph}} > 	\Findep(\bm{\vec n_\xpovm},\bm{n_K})) \leq \epssrc^2
\end{equation}
that incorporates source imperfections. In this section, we prove that the function $\Findep$ from \cite{curras-lorenzoSecurityFramework2025} is non-decreasing with respect to the number of key rounds $\bm{n_K}$, allowing us to apply \cref{maincorollary} to extend the phase-error bound to incorporate detector efficiency mismatches. To facilitate this proof, we introduce a slight modification to the original bound that makes it slightly (but strictly) tighter in the finite-key regime while preserving its asymptotic properties. Although the original bound in \cite{curras-lorenzoSecurityFramework2025} likely already satisfies the monotonicity condition, our modification not only improves the bound's tightness but also simplifies the monotonicity proof. To understand this section, the reader should be already familiar with the proof in \cite{curras-lorenzoSecurityFramework2025}.

In \cite[Eq.~(S231)]{curras-lorenzoSecurityFramework2025}, the following inequality is derived:
\begin{equation}
\label{eq:finite_formula_3alt}
   \bm{\tilde{N}_{Z_C=0}^{\rm {err}}} \leq \bm{\tilde{N}_{Z_C=0}^{\rm det}} G_+ \bigg(\frac{\bm{\tilde{N}_{Z_C=1}^{\rm {err}}}}{\bm{\tilde{N}_{Z_C=1}^{\rm det}}},	1 - \frac{2 p_{Z_C} \bm{\tilde{N}_{X_C=1}^{\rm{{det}}}}}{p_{X_C} (\bm{\tilde{N}_{Z_C=0}^{\rm det}}+\bm{\tilde{N}_{Z_C=1}^{\rm det}})}\bigg),
\end{equation}
where $0 < p_{X_C} < 1$, and all terms $\bm{\tilde N}$ represent sums of conditional expectations of Bernoulli random variables (RVs), and are thus non-negative RVs. For consistency with the rest of our paper, we denote all RVs in bold, although this convention is not used in \cite{curras-lorenzoSecurityFramework2025}. Here, the function $G_{+}$ is defined as
\begin{equation}
  \label{eq:funcG}
   			G_+(y,z) =
   			\begin{cases}
   				 y + (1-z^2)(1-2y) + 2\sqrt{z^2(1-z^2)y(1-y)}  & \quad \text{if } 0\leq y < z^2 \leq 1\text{ and } z > 0, \\
   				1  & \quad \text{otherwise.}
\end{cases}
\end{equation}
Note that the bound in \cref{eq:finite_formula_3alt} is only well-defined if $\bm{\tilde{N}_{Z_C=1}^{\rm det}} >0$; one could address this by simply defining the bound as \cref{eq:finite_formula_3alt} if $\bm{\tilde{N}_{Z_C=1}^{\rm det}} >0$ or the trivial $\bm{\tilde{N}_{Z_C=0}^{\rm {err}}} \leq \bm{\tilde{N}_{Z_C=0}^{\rm det}}$ if $\bm{\tilde{N}_{Z_C=1}^{\rm det}} =0$.

The analysis then applies Azuma's inequality to replace sums of conditional expectations with sums of Bernoulli RVs themselves, yielding bounds of the form
\begin{equation}
   \bm{\tilde N} \underset{\epsilon_A}{\leq} \bm{N} + \bm{\Delta_A} \qquad \bm{\tilde N} \underset{\epsilon_A}{\geq} \bm{N} - \bm{\Delta_A},
\end{equation}
where we are using the shorthand notation $A \underset{\epsilon}{\leq} B = \Pr[A>B] \leq \epsilon$. The original derivation applies Azuma's inequality to $\bm{\tilde{N}_{Z_C=0}^{\rm det}}+\bm{\tilde{N}_{Z_C=1}^{\rm det}}$, resulting in \cite[Eq.~(S233)]{curras-lorenzoSecurityFramework2025}:
\begin{equation}
\label{eq:old_formula}
   \bm{N_{Z_C=0}^{\rm {err}}}  \underset{6\epsilon_A}{\leq} (\bm{N_{Z_C=0}^{\rm det}} + \bm{\Delta_A}) \,   G_+ \bigg(\frac{\bm{N_{Z_C=1}^{\rm {err}}} + \bm{\Delta_A}}{\bm{N_{Z_C=1}^{\rm det}} - \bm{\Delta_A}},	1 - \frac{2 p_{Z_C} (\bm{N_{X_C=1}^{\rm{{det}}}} + \bm{\Delta_A})}{p_{X_C} (\bm{N_{Z_C=0}^{\rm det}} + \bm{N_{Z_C=1}^{\rm det}} - \bm{\Delta_A})}\bigg) + \bm{\Delta_A}.
\end{equation}

We now derive an improved bound by observing that the denominator in the second argument of $G_+$ can be tightened, that is,
\begin{equation}
\label{eq:new_formula}
   \bm{N_{Z_C=0}^{\rm {err}}}  \underset{\bluemath{5}\epsilon_A}{\leq} (\bm{N_{Z_C=0}^{\rm det}} + \bm{\Delta_A}) \,   G_+ \bigg(\frac{\bm{N_{Z_C=1}^{\rm {err}}} + \bm{\Delta_A}}{\bm{N_{Z_C=1}^{\rm det}} - \bm{\Delta_A}},	1 - \frac{2 p_{Z_C} (\bm{N_{X_C=1}^{\rm{{det}}}} + \bm{\Delta_A})}{p_{X_C} (\bm{N_{Z_C=0}^{\rm det}} + \bm{\Delta_A} + \max(0,\bm{N_{Z_C=1}^{\rm det}} - \bm{\Delta_A}))}\bigg) + \bm{\Delta_A}.
\end{equation}
This bound is tighter since the RHS of Eq.~\eqref{eq:new_formula} is smaller to that of Eq.~\eqref{eq:old_formula}\blue{, and the failure probability of the bound is also smaller.} Again, the RHS of \cref{eq:old_formula,eq:new_formula} are only well-defined if $\bm{N_{Z_C=1}^{\rm det}} - \bm{\Delta_A} \geq 0$; if not, one can trivially replace it by $\bm{N_{Z_C=0}^{\rm {det}}}$.

Let us define the function
\begin{equation}
\label{eq:def_f}
 f(x;y,a,c,p)
 \;:=\;
 x\,G_{+}\!\Bigl(y,\;1-\tfrac{a}{p(x+c)}\Bigr),
 \qquad (x,y,a,c,p)\in\mathcal D.
\end{equation}
with domain
\begin{equation} 
\mathcal D
 :=\bigl\{(x,y,a,c,p)\in\mathbb R^{5}\,\bigm|\,
        x\geq 0,c\geq 0,a\geq 0,\; 0\leq y \leq 1,\;0<p<1\bigr\}.
\end{equation}
By identifying
\begin{equation}
 x   = \bm{\tilde N_{Z_C=0}^{\rm det}}, \qquad
 y   = \tfrac{\bm{\tilde N_{Z_C=1}^{\rm err}}}{\bm{\tilde N_{Z_C=1}^{\rm det}}},\qquad
 a   = 2p_{Z_C}\bm{\tilde N_{X_C=1}^{\rm det}},\qquad
 c   = \bm{\tilde N_{Z_C=1}^{\rm det}},\qquad
 p   = p_{X_C},
\end{equation}
we can rewrite the bound in \cref{eq:finite_formula_3alt} as
\begin{equation}
\label{eq:finite_formula_3alt_f}
 \bm{\tilde N_{Z_C=0}^{\rm err}}
 \;\le\;
 f\bigl(x;y,a,c,p\bigr).
\end{equation}
The function $f(x;y,a,c,p)$ is non-decreasing in $x$, as proven in \cref{lem:monotonicity} in \cref{subapp:proof_monotonicity}. Using this and the fact that $G_+(y,z)$ is non-decreasing with $y$ and non-increasing with $z$, one can directly obtain \cref{eq:new_formula}. \blue{Note that the failure probability in \cref{eq:new_formula} is $5 \epsilon_A$ rather than $6 \epsilon_A$ as in \cref{eq:old_formula}. This is because, in \cref{eq:old_formula}, Azuma's inequality was applied to obtain a lower bound on $\bm{\tilde{N}_{Z_C=0}^{\rm det}}+\bm{\tilde{N}_{Z_C=1}^{\rm det}}$, and this application contributes $\epsilon_A$ to the overall failure probability. Conversely, in the new bound, Azuma's inequality is applied independently to obtain an upper bound on $\bm{\tilde{N}_{Z_C=0}^{\rm det}}$ and a lower bound on $\bm{\tilde{N}_{Z_C=1}^{\rm det}}$, respectively. However, since these bounds were already used in the formula, they do not contribute again to the overall failure probability.}



In \cite[Eq.~(S116)]{curras-lorenzoSecurityFramework2025}, it is shown that, by applying the Chernoff bound and Hoeffding bound to \cref{eq:old_formula}, one can obtain the following bound on the number of phase errors $\bm{\nph}$ (called $N_{\rm ph}$ in \cite{curras-lorenzoSecurityFramework2025}):
\begin{equation}
\label{eq:Nph_bound_finite_old}
\begin{aligned}
\bm{\nph} &\underset{\epsilon_{\rm tot}}{\leq} \frac{1}{p_{Z_C} p_{\texttt{TAR}\vert \textrm{vir}}} \Bigg[(\bm{\overline N_{Z_C=0}^{\rm det}} + \bm{\Delta_A}) \,   G_+ \bigg(\frac{\bm{\overline N_{Z_C=1}^{\rm {err}}} + \bm{\Delta_A}}{\bm{\underline N_{Z_C=1}^{\rm det}} - \bm{\Delta_A}},	1 - \frac{2 p_{Z_C} (\bm{\overline N_{X_C=1}} + \bm{\Delta_A})}{p_{X_C} (\bm{\underline N_{Z_C=0}^{\rm det}} + \bm{\underline N_{Z_C=1}^{\rm det}} - \bm{\Delta_A})}\bigg) \\
&\,- \sum_{\alpha \in \{0,1\}} \sum_{j \in \mathcal{C}_{\rm Tar}^{(\alpha)}} \bm{\underline N_{j,\texttt{TAR}\alpha,Z_C}^{(\alpha \oplus 1)_X}} + \bm{\Delta_A} + \bm{\Delta_H}\Bigg],
\end{aligned}
\end{equation}
where $\epsilon_{\rm tot} \coloneqq 6 \epsilon_A + (2\upsilon_{r}+3\upsilon_t+3)\epsilon_C + \epsilon_H$ and $p_{\texttt{TAR}\vert \textrm{vir}} > 0$. However, if we apply the same steps starting from our improved bound in \cref{eq:new_formula}, we obtain the tighter bound
\begin{equation}
\label{eq:Nph_bound_finite_new}
\begin{aligned}
\bm{\nph} &\underset{\epsilon'_{\rm tot}}{\leq} \frac{1}{p_{Z_C} p_{\texttt{TAR}\vert \textrm{vir}}} \Bigg[(\bm{\overline N_{Z_C=0}^{\rm det}}  + \bm{\Delta_A} ) \,   G_+ \bigg(\frac{\bm{\overline N_{Z_C=1}^{\rm {err}}} + \bm{\Delta_A} }{\bm{\underline N_{Z_C=1}^{\rm det}} - \bm{\Delta_A} },	1 - \frac{2 p_{Z_C} (\bm{\overline N_{X_C=1}} + \bm{\Delta_A} )}{p_{X_C} (\bm{\overline N_{Z_C=0}^{\rm det}} + \bm{\Delta_A}  +  \max(\bm{\underline N_{Z_C=1}^{\rm det}} - \bm{\Delta_A} ,0))}\bigg) \\
&\,- \sum_{\alpha \in \{0,1\}} \sum_{j \in \mathcal{C}_{\rm Tar}^{(\alpha)}} \bm{\underline N_{j,\texttt{TAR}\alpha,Z_C}^{(\alpha \oplus 1)_X}} + \bm{\Delta_A}  + \bm{\Delta_H} \Bigg],
\end{aligned}
\end{equation}
where $\epsilon'_{\rm tot} \coloneqq \bluemath{5} \epsilon_A + (2\upsilon_{r}+2\upsilon_t+3)\epsilon_C + \epsilon_H$. Note that \blue{in $\epsilon_{\rm tot}$, we had a $(2\upsilon_{r}+3\upsilon_t+3)\epsilon_C$ term, whereas in $\epsilon'_{\rm tot}$, the term is now reduced to $(2\upsilon_{r}+2\upsilon_t+3)\epsilon_C$. This is due to the fact that}  $\bm{\underline N_{Z_C=0}^{\rm det}}$ no longer appears in \cref{eq:Nph_bound_finite_new}, and therefore the concentration inequalities used to bound this term in \cite{curras-lorenzoSecurityFramework2025}\blue{, which contributed $\upsilon_t \epsilon_C$ to the overall failure probability,} no longer need to be used. Again, the RHS of \cref{eq:Nph_bound_finite_old,eq:Nph_bound_finite_new} are only well-defined when $\bm{\underline N_{Z_C=1}^{\rm det}} - \bm{\Delta_A} \geq 0$; otherwise, we can replace the RHS by $\bm{n_K}$, since the trivial bound $\bm{\nph} \leq \bm{n_K}$ holds deterministically.

All RVs on the RHS of Eq.~\eqref{eq:Nph_bound_finite_new} are functions of the data observed and announced in the actual protocol. We want to prove that the RHS of Eq.~\eqref{eq:Nph_bound_finite_new} is increasing with respect to the number of key rounds $\bm{n_K}$ (denoted by $N_{\rm key}^{\rm det}$ in \cite{curras-lorenzoSecurityFramework2025}). For this, we note that the only terms that depend on $\bm{n_K}$ are $\bm{N_{Z_C=0}^{\rm det}}$, $\bm{\Delta_A}$, and $\bm{\Delta_H}$, all of which are non-decreasing with respect to $\bm{n_K}$. All other RVs depend only on the data from rounds in which Bob used the $\xpovm$ POVM, which we denote by the random vector $\bm{\vecnX}$. To indicate the RVs that depend on $\bm{n_K}$, we write them as explicit functions of $\bm{n_K}$:
\begin{equation}
\label{eq:Nph_bound_finite_new_dep}
\begin{aligned}
\bm{\nph} &\underset{\epsilon'_{\rm tot}}{\leq} \frac{1}{p_{Z_C} p_{\texttt{TAR}\vert \textrm{vir}}} \Bigg[\big(\bm{\overline N_{Z_C=0}^{\rm det}} (\bm{n_K}) + \bm{\Delta_A} (\bm{n_K})\big) \,   G_+ \Bigg(\frac{\bm{\overline N_{Z_C=1}^{\rm {err}}} + \bm{\Delta_A} (\bm{n_K})}{\bm{\underline N_{Z_C=1}^{\rm det}} - \bm{\Delta_A} (\bm{n_K})}, \\
&~ 1 - \frac{2 p_{Z_C} \big(\bm{\overline N_{X_C=1}} + \bm{\Delta_A} (\bm{n_K})\big)}{p_{X_C} \Big(\bm{\overline N_{Z_C=0}^{\rm det}}(\bm{n_K}) + \bm{\Delta_A} (\bm{n_K}) +  \max\big(\bm{\underline N_{Z_C=1}^{\rm det}} - \bm{\Delta_A} (\bm{n_K}),0\big)\Big)}\Bigg) \\
&~- \sum_{\alpha \in \{0,1\}} \sum_{j \in \mathcal{C}_{\rm Tar}^{(\alpha)}} \bm{\underline N_{j,\texttt{TAR}\alpha,Z_C}^{(\alpha \oplus 1)_X}} + \bm{\Delta_A} (\bm{n_K}) + \bm{\Delta_H}(\bm{n_K}) \Bigg].
\end{aligned}
\end{equation}

To analyze the monotonicity systematically, we define the auxiliary functions
\begin{equation}
   g(x,z) = x \,   G_+ \bigg(\frac{\bm{\overline{N}_{Z_C=1}^{\rm {err}}} + \bm{\Delta_A}(z)}{\bm{\underline{N}_{Z_C=1}^{\rm det}} - \bm{\Delta_A}(z)},	1 - \frac{2 p_{Z_C} (\bm{\overline{N}_{X_C=1}} + \bm{\Delta_A}(z))}{p_{X_C} (x + \max(0,\bm{\underline{N}_{Z_C=1}^{\rm det}} - \bm{\Delta_A}(z)))}\bigg) + \bm{\Delta_A}(z) +  \bm{\Delta_H}(z)
\end{equation}
and
\begin{equation}
   h(z) = \bm{\overline{N}_{Z_C=0}^{\rm det}}(z) + \bm{\Delta_A} (z).
\end{equation}

Using these definitions, we can rewrite the bound in Eq.~\eqref{eq:Nph_bound_finite_new_dep} as
\begin{equation}
\label{eq:Nph_bound_finite_new_dep_rewritten}
\begin{aligned}
\bm{\nph} &\underset{\epsilon'_{\rm tot}}{\leq} \frac{1}{p_{Z_C} p_{\texttt{TAR}\vert \textrm{vir}}} \Bigg[g(h(\bm{n_K}), \bm{n_K}) - \sum_{\alpha \in \{0,1\}} \sum_{j \in \mathcal{C}_{\rm Tar}^{(\alpha)}} \bm{\underline N_{j,\texttt{TAR}\alpha,Z_C}^{(\alpha \oplus 1)_X}}\Bigg].
\end{aligned}
\end{equation}
Clearly, if we prove that $g(h(\bm{n_K}), \bm{n_K})$ is non-decreasing with respect to $\bm{n_K}$, then it follows that the RHS of Eq.~\eqref{eq:Nph_bound_finite_new_dep_rewritten} also satisfies this property. To prove this, it suffices to show that:
\begin{itemize}
    \item $g(x,z)$ is non-decreasing with respect to both $x$ and $z$, and
    \item $h(z)$ is non-decreasing with respect to $z$.
\end{itemize}

The latter follows directly from the fact that both $\bm{\overline{N}_{Z_C=0}^{\rm det}}(\bm{n_K})$ and $\bm{\Delta_A} (\bm{n_K})$ are non-decreasing with respect to $\bm{n_K}$. Moreover, $g(x,z)$ is non-decreasing with $z$ due to the fact that $\bm{\Delta_A} (\bm{n_K})$ and $\bm{\Delta_H} (\bm{n_K})$ are non-decreasing with $\bm{n_K}$, and the fact that the function $G_+(y,z)$ is non-decreasing with $y$ and non-increasing with $z$.

Thus, it only remains to prove that $g(x,z)$ is non-decreasing with $x$. Note that we can write
\begin{equation}
   g(x,z) = f(x; y, a, c, p) + \bm{\Delta_A}(z) +  \bm{\Delta_H}(z)
\end{equation}
by identifying
\begin{equation}
y   = \tfrac{\bm{\overline{N}_{Z_C=1}^{\rm {err}}} + \bm{\Delta_A}(z)}{\bm{\underline{N}_{Z_C=1}^{\rm det}} - \bm{\Delta_A}(z)},\qquad
a   = 2p_{Z_C}(\bm{\overline{N}_{X_C=1}} + \bm{\Delta_A}(z)),\qquad
c   = \max(0,\bm{\underline{N}_{Z_C=1}^{\rm det}} - \bm{\Delta_A}(z)),\qquad
p   = p_{X_C}.
\end{equation}

Therefore, the fact that $g(x,z)$ is non-decreasing with $x$ follows from Lemma~\ref{lem:monotonicity}. Note that the monotonicity proof requires that $y \geq 0$, which corresponds to the condition $\bm{\underline{N}_{Z_C=1}^{\rm det}} - \bm{\Delta_A}(z) > 0$ (since the numerator $\bm{\overline{N}_{Z_C=1}^{\rm {err}}} + \bm{\Delta_A}(z) \geq 0$ by definition). When this condition holds, we can apply Lemma~\ref{lem:monotonicity} to establish that $g(x,z)$ is non-decreasing with respect to $x$, and hence the overall bound is non-decreasing with respect to $\bm{n_K}$. 

However, if $\bm{\underline{N}_{Z_C=1}^{\rm det}} - \bm{\Delta_A}(z) \leq 0$, then the bounds in Eqs.~\eqref{eq:Nph_bound_finite_new} and \eqref{eq:Nph_bound_finite_new_dep} are not well-defined. In this case, as explained earlier, we can resort to the trivial bound $\bm{\nph} \underset{\epsilon'_{\rm tot}}{\leq} \bm{n_K}$, which is obviously non-decreasing with respect to $\bm{n_K}$. Since this bound holds deterministically, it certainly holds with failure probability $\epsilon'_{\rm tot}$. 

Therefore, by identifying $\epssrc^2 = \epsilon'_{\rm tot}$, we have proven that the phase-error bound can be written as in \cref{eq:nph_bound}, where $\Findep$ is non-decreasing with respect to the number of key rounds $\bm{n_K}$, completing the proof.

\subsection{Proof of \texorpdfstring{\cref{lem:monotonicity}}{Lemma 1}}
\label{subapp:proof_monotonicity}

\begin{lemma}
\label{lem:monotonicity}
Consider the function $f(x) \equiv f(x;y,a,c,p)$ for any fixed value of $(y,a,c,p)$ such that $(0,y,a,c,p) \in \mathcal{D}$, where $f(x;y,a,c,p)$ is defined in \cref{eq:def_f}. The function $f(x)$ is non-decreasing in $x$ for $x \geq 0$.
 \end{lemma}

\begin{proof}
Let $ z(x) = 1 - \frac{a}{p(x+c)}$, such that $f(x) = G_+(y,z(x))$. We establish that $f(x)$ is non-decreasing in $x$ by proving (i) continuity of $f(x)$ for all $x \geq 0$, and (ii) non-negativity of $f'(x)$ wherever the derivative exists. 

\subsubsection*{1: Continuity of \texorpdfstring{$f(x)$ for $x \geq 0$}{f(x) for x >= 0}.}

The function $f$ is clearly continuous everywhere except possibly at points where the definition of $G_+(y,z(x))$ changes. These occur when $z(x)$ crosses the boundaries of the region $\{0 \leq y < z^2 \leq 1 \text{ and } z > 0\}$. Note that $f(0) = 0$ by direct evaluation, so we consider only $x > 0$.

\begin{enumerate}
  \item \textit{Boundary $z(x_b) = 0$.} Let $x_b > 0$ satisfy $z(x_b) = 0$.
  \begin{itemize}
    \item As $x \to x_b^-$: We have $z(x) < 0$, so we are in the ``otherwise'' branch where $G_+(y,z(x)) = 1$. Thus $f(x) = x$ and $\lim_{x \to x_b^-} f(x) = x_b$.
    
    \item At $x = x_b$: Since $z(x_b) = 0$ fails the condition $z > 0$, we have $G_+(y,z(x_b)) = 1$ and $f(x_b) = x_b$.
    
    \item As $x \to x_b^+$: We have $z(x) \to 0^+$. For any $y > 0$, there exists a right-neighborhood of $x_b$ where $z(x)^2 < y$, and therefore $G_+(y,z(x)) = 1$ and $f(x) = x$. Hence $\lim_{x \to x_b^+} f(x) = x_b$. (For $y = 0$, the main branch gives $G_+(0,z) = 1 - z^2 \to 1$ as $z \to 0^+$, maintaining continuity.)
  \end{itemize}
  
  \item \textit{Boundary $z(x_b)^2 = y$.} Let $x_b > 0$ satisfy $z(x_b) = \sqrt{y}$ for $y \in (0,1)$.
  \begin{itemize}
    \item As $x \to x_b^-$: We have $z(x)^2 < y$, so $G_+(y,z(x)) = 1$ and $f(x) = x$. Thus $\lim_{x \to x_b^-} f(x) = x_b$.
    
    \item As $x \to x_b^+$: We have $z(x) > \sqrt{y}$, entering the main branch. Computing the limit:
    \begin{align}
      \lim_{z \to \sqrt{y}^+} G_+(y,z) &= y + (1-y)(1-2y) + 2\sqrt{y(1-y) \cdot y(1-y)} \notag\\
      &= y + (1-y)(1-2y) + 2y(1-y) \notag\\
      &= y + 1 - 3y + 2y^2 + 2y - 2y^2 = 1.
    \end{align}
    Therefore, $\lim_{x \to x_b^+} f(x) = x_b$.
  \end{itemize}
  
  \item \textit{Boundary $z(x_b)^2 = 1$.} This requires $z(x_b) = 1$, which occurs only when $a = 0$. In this case, $z(x) = 1$ for all $x$, obtaining $f(x) = xy$ (if $y < 1$) or $f(x) = x$ (if $y = 1$), both continuous.
\end{enumerate}

Since $f(x)$ is continuous at all transition points, it is continuous on $[0,\infty)$.

\subsubsection*{2: Non-negativity of \texorpdfstring{$f'(x)$}{f'(x)}.}

We examine the derivative on open intervals where it exists:

\begin{enumerate}
  \item \textit{Case 1: $G_+(y,z(x)) = 1$.} This occurs when $z(x) \leq 0$ or when $0 < z(x) \leq 1$ with $z(x)^2 \leq y$. Here $f(x) = x$, so $f'(x) = 1 > 0$.
  
  \item \textit{Case 2: $G_+(y,z(x))$ follows the main branch.} This occurs when $0 \leq y < z(x)^2 \leq 1$ and $z(x) > 0$. The derivative is
  \begin{equation}
    f'(x) = G_+(y,z(x)) + x \frac{\partial G_+}{\partial z}(y,z(x)) \cdot z'(x).
  \end{equation}
    A symbolic check in \texttt{Mathematica} under $(x,y,a,c,p)\in\mathcal D$ and $z(x) > 0$ confirms that this derivative is never negative:
    \begin{verbatim}
    In[1]:= 
    Gplus[y_, z_] := y + (1 - z^2)*(1 - 2 y) + 2 Sqrt[z^2*(1 - z^2)*y*(1 - y)];
    In[2]:= z[x_] := 1 - a/(p*x + c);
    In[3]:= f[x_] := x*Gplus[y, z[x]];
    In[4]:= df = Simplify[D[f[x], x]];
    In[5]:= Reduce[{df < 0, x >= 0, 0 <= y <= 1, c >= 0, 0 <= p <= 1, 
     a >= 0, z[x] > 0}]
    Out[5]= False
    \end{verbatim}
\end{enumerate}

\noindent Since $f(x)$ is continuous on $[0,\infty)$ with $f'(x) \geq 0$ wherever the derivative exists, it follows that $f(x)$ is non-decreasing in $x$.
\end{proof}

\section{Additional simulation results}
\label{app:additional_simulations}

\blue{Here, we provide additional simulations of the secret-key rate obtainable when applying our modular lifting theorem to the security proof in Ref.~\cite{curras-lorenzoSecurityFramework2025}. In \cref{fig:graph2_varying_delta}, we investigate the effect of varying the parameter $\Delta$, which quantifies the detector efficiency mismatch, while keeping the rest of the parameters fixed. Conversely, in \cref{fig:graph3_varying_epsilon}, we investigate the effect of varying the parameter $\epsilon$, which quantifies the magnitude of uncharacterized source imperfections and side-channels, while keeping the rest of the parameters fixed. We have not included graphs varying the parameter $\delta_\mathrm{SPF}$, which quantifies characterized qubit state preparation flaws, since this is parameter has a very minor effect on the key rate \cite{curras-lorenzoSecurityFramework2025}.}

\begin{figure}[h]
(a)
\includegraphics[width=8.3cm]{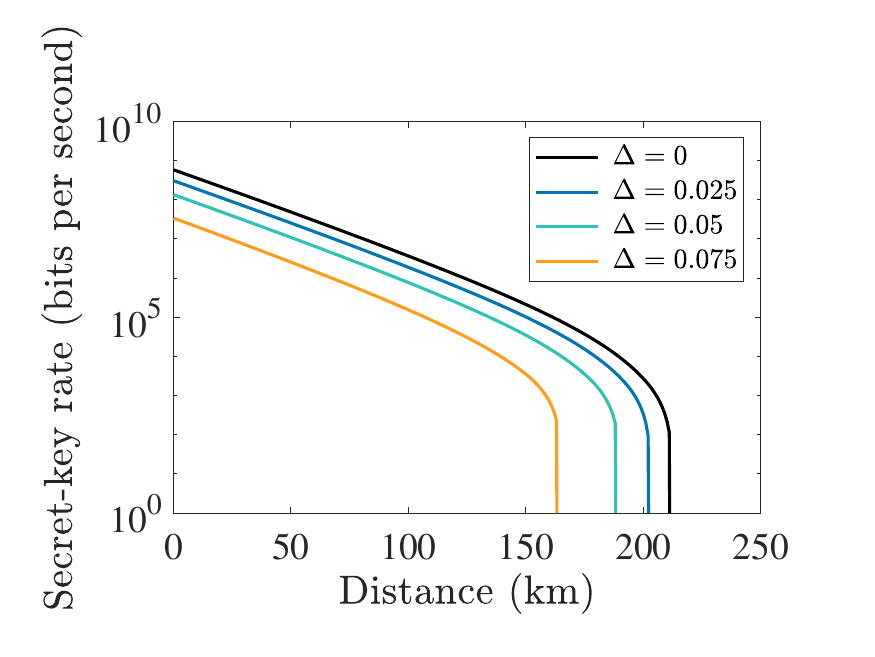}
(b)
\includegraphics[width=8.3cm]{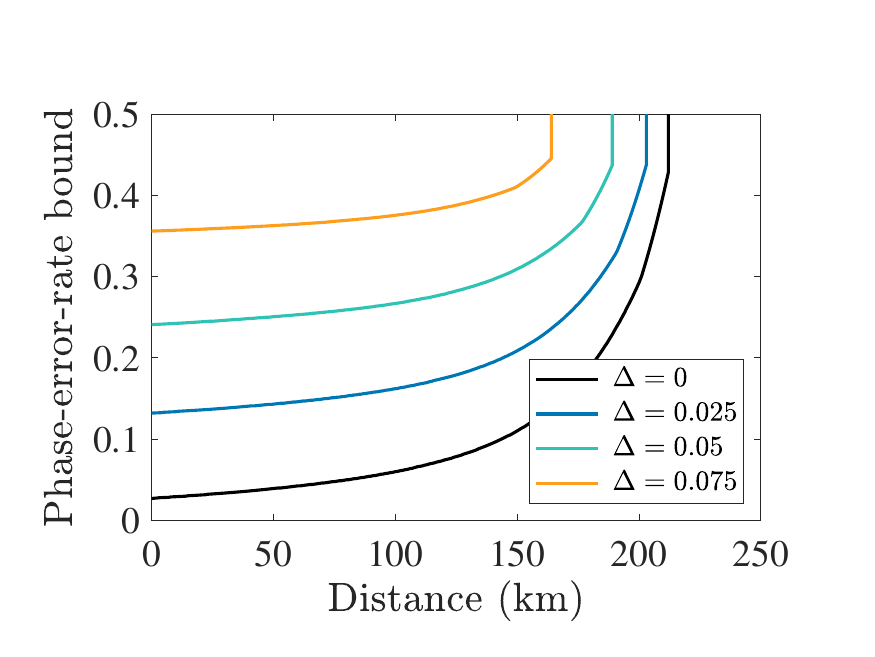}

	\caption{\blue{(a) Finite-size secret-key rate for a BB84-type protocol as a function of $\Delta$, which quantifies the magnitude of the detection efficiency mismatch. We consider $N = 10^{12}$, $\epsilon = 10^{-6}$ and $\delta_\mathrm{SPF} = 0.063$. (b) Corresponding bound on the phase-error rate.}}
    \label{fig:graph2_varying_delta}
\end{figure}

\begin{figure}[h]
(a)
\includegraphics[width=8.3cm]{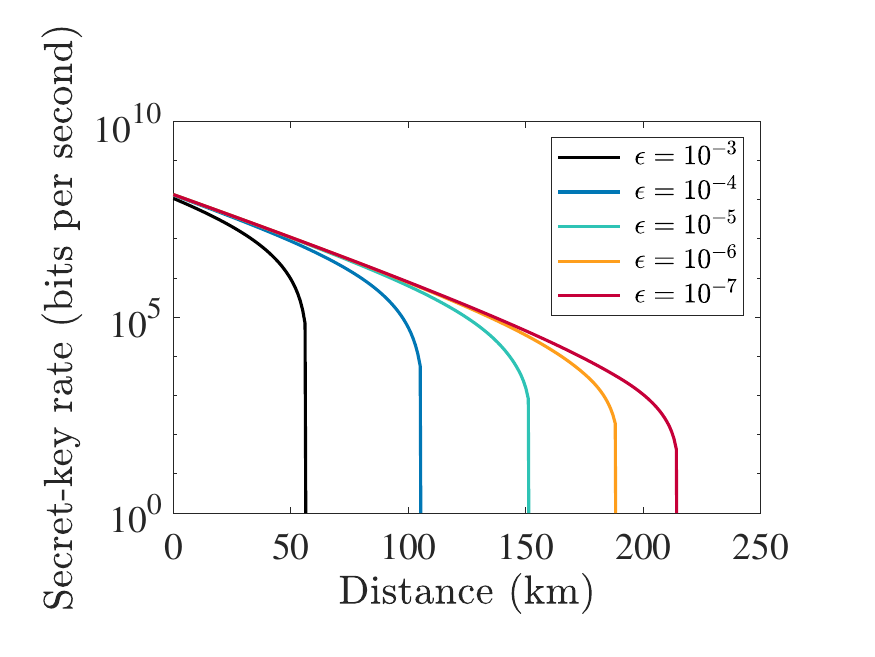}
(b)
\includegraphics[width=8.3cm]{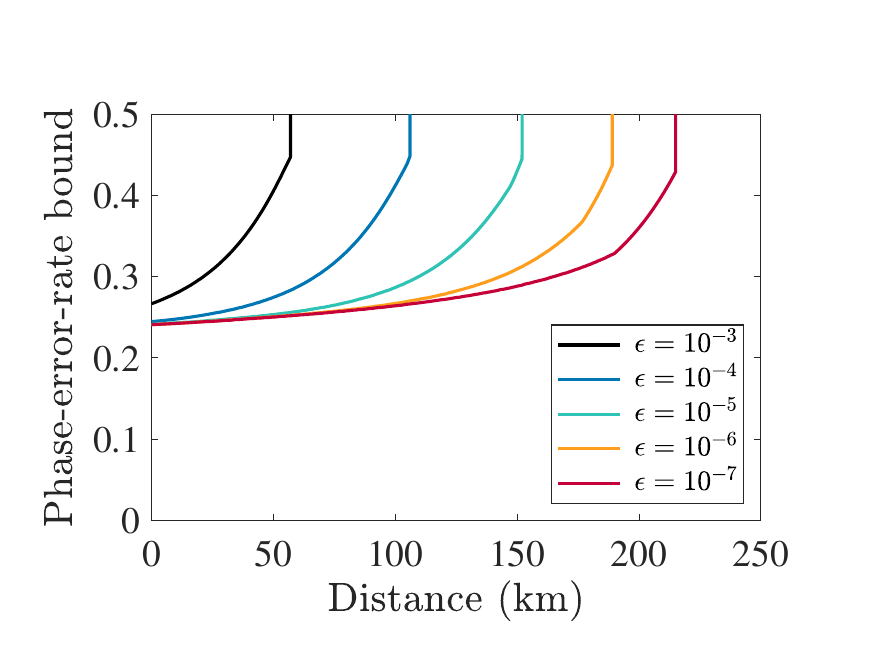}

	\caption{\blue{(a) Finite-size secret-key rate for a BB84-type protocol as a function of $\epsilon$, which quantifies the magnitude of uncharacterized source imperfections and side-channels. We consider $N = 10^{12}$, $\Delta = 0.05$ and $\delta_\mathrm{SPF} = 0.063$. (b) Corresponding bound on the phase-error rate.}}
    \label{fig:graph3_varying_epsilon}
\end{figure}

\blue{Comparing \cref{fig:graph2_varying_delta,fig:graph3_varying_epsilon}, it is clear that $\Delta$ and $\epsilon$ have quite different effects on the key rate. The former has an impact even at distance $\SI{0}{\kilo\metre}$, and then its impact does not increase much with distance. In contrast,  $\epsilon$ has a very minor impact at $\SI{0}{\kilo\metre}$, but its impact increases dramatically as the distance grows. The reason for the latter behaviour is explained in the main text.}

\blue{In \cref{fig:graph4_varying_N_with_imperfections}, we investigate the impact of finite-size effects by varying the total number of transmitted rounds $N$. We can see that finite-key effects have a significant impact, especially at low values of $N$ such as $N=10^{8}$. We remark that, for the most part, this impact is due to the phase-error bound for source imperfections in Ref.~\cite{curras-lorenzoSecurityFramework2025}, not from the modular theorem introduced in this work to incorporate detection efficiency mismatches. This becomes clear by comparing it with \cref{fig:graph5_varying_N_ideal}, in which we consider an ideal BB84 source and apply our modular lifting theorem to the standard BB84 phase-error bound for such an ideal source.}

\begin{figure}[H]	
\centering
\includegraphics[width=8.5cm]{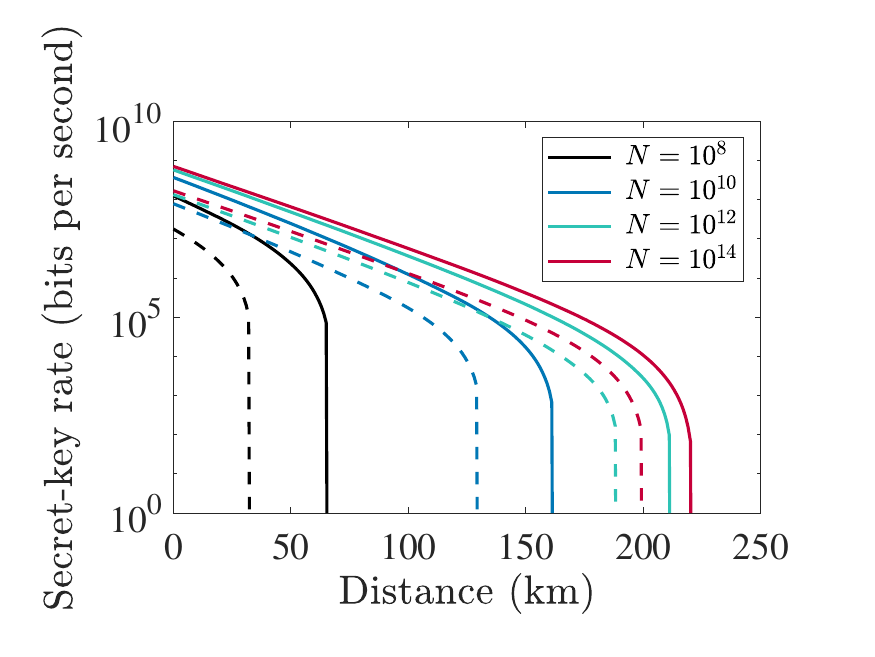} 
	\caption{\blue{Finite-size secret-key rate for a BB84-type protocol as a function of the total number of transmitted rounds $N$. We consider source imperfection parameters $\epsilon = 10^{-6}$ and $\delta_\mathrm{SPF} = 0.063$. The solid lines correspond to the original analysis in Refs.~\cite{curras-lorenzoSecurityFramework2025} for $\Delta = 0$ (i.e., no detector efficiency mismatch), while the dashed lines correspond to applying our modular lifting theorem to Ref.~\cite{curras-lorenzoSecurityFramework2025} to incorporate a detector efficiency mismatch with $\Delta = 0.05$.}}
\label{fig:graph4_varying_N_with_imperfections}
\end{figure}

\begin{figure}[H]
\centering

\includegraphics[width=8.5cm]{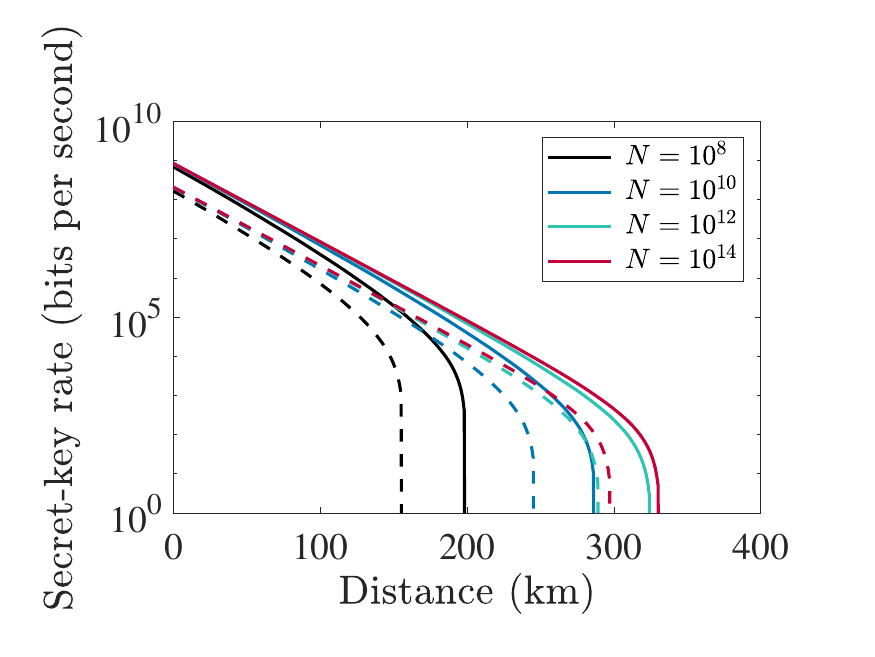} 
	\caption{\blue{Finite-size secret-key rate for a BB84-type protocol as a function of the total number of transmitted rounds $N$. We consider an ideal BB84 source, i.e., $\epsilon = 0$ and $\delta_\mathrm{SPF} = 0$. The solid lines correspond to the standard BB84 phase-error analysis for $\Delta = 0$ (i.e., no detector efficiency mismatch), while the dashed lines correspond to applying our modular lifting theorem to standard BB84 phase-error analysis to incorporate a detector efficiency mismatch with $\Delta = 0.05$.}}
\label{fig:graph5_varying_N_ideal}
\end{figure}

\end{widetext}